\documentclass[review]{elsarticle}
\usepackage{amsmath,amssymb}
\usepackage{amsthm}
\usepackage{lineno,hyperref}
\usepackage{xcolor}
\usepackage[ruled,vlined, linesnumbered]{algorithm2e}
\usepackage{float}
\modulolinenumbers[5]
\journal{Journal of \LaTeX\ Templates}

\usepackage{dsfont}
\usepackage{todonotes}
\newcommand{\TAU}{\Large{\boldsymbol{\tau}}}
\DeclareMathOperator*{\mintheta}{min}
\DeclareMathOperator*{\argmin}{arg\,min}

\newcommand{\FPOP}{\textbf{FPOP}} 
\newcommand{\PELT}{\textbf{PELT}} 
\newcommand{\OP}{\textbf{OP}} 
\newcommand{\PDPA}{\textbf{PDPA}} 
\newcommand{\PhiFPOP}{\textbf{Ms.FPOP}} 
\newcommand{\PhiPELT}{\textbf{Ms.PELT}}
\newcommand{\Future}{\mathcal{A}} 
\newcommand{\Past}{\mathcal{B}} 
\newtheorem{lemma}{Lemma}
\newtheorem{assumption}{Assumption}
\hypersetup{pdfauthor={Name}}
\date{}

\begin{document}
\begin{frontmatter}
\title{\PhiFPOP: An Exact and Fast Segmentation Algorithm
With a Multiscale Penalty}

\cortext[corresp_author]{Corresponding author}
\author[evryPS,ips2]{Liehrmann Arnaud\corref{corresp_author}}
\ead{arnaud.liehrmann@universite-paris-saclay.fr}
\author[evryPS,ips2]{Rigaill Guillem}
\address[evryPS]{Universit\'e Paris-Saclay, CNRS, Univ Evry, Laboratoire de Math\'ematiques et Mod\'elisation d'Evry, 91037, Evry-Courcouronnes, France.}
\address[ips2]{Universit\'e Paris-Saclay, CNRS, INRAE, Universit\'e  Evry, Institute of Plant Sciences Paris-Saclay (IPS2), 91190, Gif sur Yvette, France.
}

\begin{abstract}
Given a time series in $\mathbb{R}^n$ with a piecewise constant mean and independent noises, we propose an exact dynamic programming algorithm to minimize a least square criterion with a multiscale penalty promoting well-spread changepoints. Such a penalty has been proposed in Verzelen \textit{et al.}~(2020), and it achieves optimal rates for changepoint detection and changepoint localization.

Our proposed algorithm, named \PhiFPOP{},  extends functional pruning ideas of Rigaill~(2015) and Maidstone \textit{et al.}~(2017) to multiscale penalties. For large signals, $n \geq 10^5$, with relatively few real changepoints, \PhiFPOP{} is typically quasi-linear and an order of magnitude faster than \PELT{}. We propose an efficient C++ implementation interfaced with R of \PhiFPOP{} allowing to segment a profile of up to $n=10^6$ in a matter of seconds.

Finally, we illustrate on simple simulations that for large enough profiles ($n\geq 10^4$) \PhiFPOP{} using the multiscale penalty of Verzelen \textit{et al.}~(2020) is typically more powerfull than \FPOP{} using the classical BIC penalty of Yao 1989.
\end{abstract}


\end{frontmatter}

\section{Introduction}

A National Research Council report \cite{national2013frontiers} identifies changepoint detection as one of the ``inferential giants'' in massive data analysis. 
Detecting changepoints, whether a posteriori or online, is important in areas as diverse as
 bioinformatics \cite{olshen2004circular,Picard2005}, econometrics and finance \cite{bai2003computation,Thies2018}, climate  \cite{Reeves2007}, autonomous
driving \cite{galceran2017multipolicy}, computer vision \cite{ranganathan2012pliss} and neuroscience \cite{jewell2020fast}.
The most common and prototypical changepoint detection problem is that of detecting changes in mean of a univariate gaussian signal :
\begin{equation}\label{eq:piecewisemodel}
y_t = f_t + \varepsilon_t, \quad \text{for } t=1, \ldots, n,
\end{equation}
where $f_t$ is a deterministic piecewise constant with changepoints  whose number $D$ and locations, $0 < \tau_1 < \ldots < \tau_D < n$, are unknown, and $\varepsilon_t$ are independant and follow a Gaussian distribution of mean 0 and variance 1. A large number of approaches have been proposed to solve this problem (amongst many others \cite{Yao,emilie,harchaoui2010multiple,frick,fryzlewicz2020detecting}, see \cite{aminikhanghahi2017survey,Truong2020} for a review).

Recently, \cite{verzelen2020optimal} characterize optimal rates for changepoint detection and changepoint localization and proposed
a least-squares estimator with a multiscale penalty achieving these optimal rates. 
This multiscale penalty depends on minus the log-length of the segments which promotes well spread changepoints. It can be written as :
\begin{equation}\label{eq:criteria_verzelenetal}
\sum_{d=1}^{D+1} \gamma + \beta \log(n) - \beta \log(\tau_d - \tau_{d-1}),
\end{equation}
where $\gamma=q L$ and $\beta=2L$ with $q$ positive and $L>1$, and with the convention that $\tau_0=0$ and $\tau_{D+1}=n$.

Up to a multiplicative constant this penalty is always smaller than the BIC penalty ($2\log(n)$) \cite{Yao}.
Intuitively, it favors balanced segmentation as:
\begin{itemize}
\item the penalty of a fixed sized segment ($r$) increases with $n$ : $\beta \log(n/r).$
\item while the penalty for a segment whose size is proportionnal to $n$ ($\alpha.n$) is constant of $n$ :  
$\beta \log(1/\alpha).$
\end{itemize}

\paragraph{Contribution}
In this paper, we propose a dynamic programming algorithm, named \PhiFPOP{} optimizing a slightly more general penalty. where the $\log(\tau_d - \tau_{d-1})$ is replaced by  $g(\tau_d - \tau_{d-1})$ for an arbitrary function $g$ satistying assumption A\ref{assumpt1}.

\paragraph{Existing works}
\PhiFPOP{} extends functional pruning techniques as in \PDPA{} or \FPOP{} \cite{PDPA,FPOP} to the case of multiscale penalties. A key condition for \FPOP{} and \PDPA{} is that the cost function is point additive (condition C1 in \cite{FPOP}). As we will explain in more details later, this condition is not verified for the multiscale penalty \eqref{eq:criteria_verzelenetal}, making the extension not trivial. The key idea behind functionnal pruning is to store the set of parameter values for which a particular change is optimal. For a classical penalty (i.e. with a point additive cost function) this set gets smaller with every new datapoint. This is not the case with the multiscale penalty making the update more complex. A key insight of \PhiFPOP{} is to store a slightly larger set that is easy to update.

Importantly, it is possible to optimize the multiscale criteria of \cite{verzelen2020optimal} using inequality based pruning as in \PELT{}. We will call \PhiPELT{} this strategy. However for large signals with relatively few true changepoints  it is our experience that \PhiPELT{} is quadratic while \PhiFPOP{} is quasi-linear. For example it can be seen on Figure \ref{fig:simple_runtime}.A that it takes about 193 seconds for \PhiPELT{} to process a signal of size $n=128000$ without any changepoint. In the same amount of time \PhiFPOP{} can process signals of size larger than $n=4\times10^6.$ In the presence of true changepoints, (one every thousand datapoints) \PhiPELT{} as expected is much faster but still slower than \PhiFPOP{} (see Figure \ref{fig:simple_runtime}.B).

\begin{figure}[H]
\centering
\includegraphics[scale=0.34]{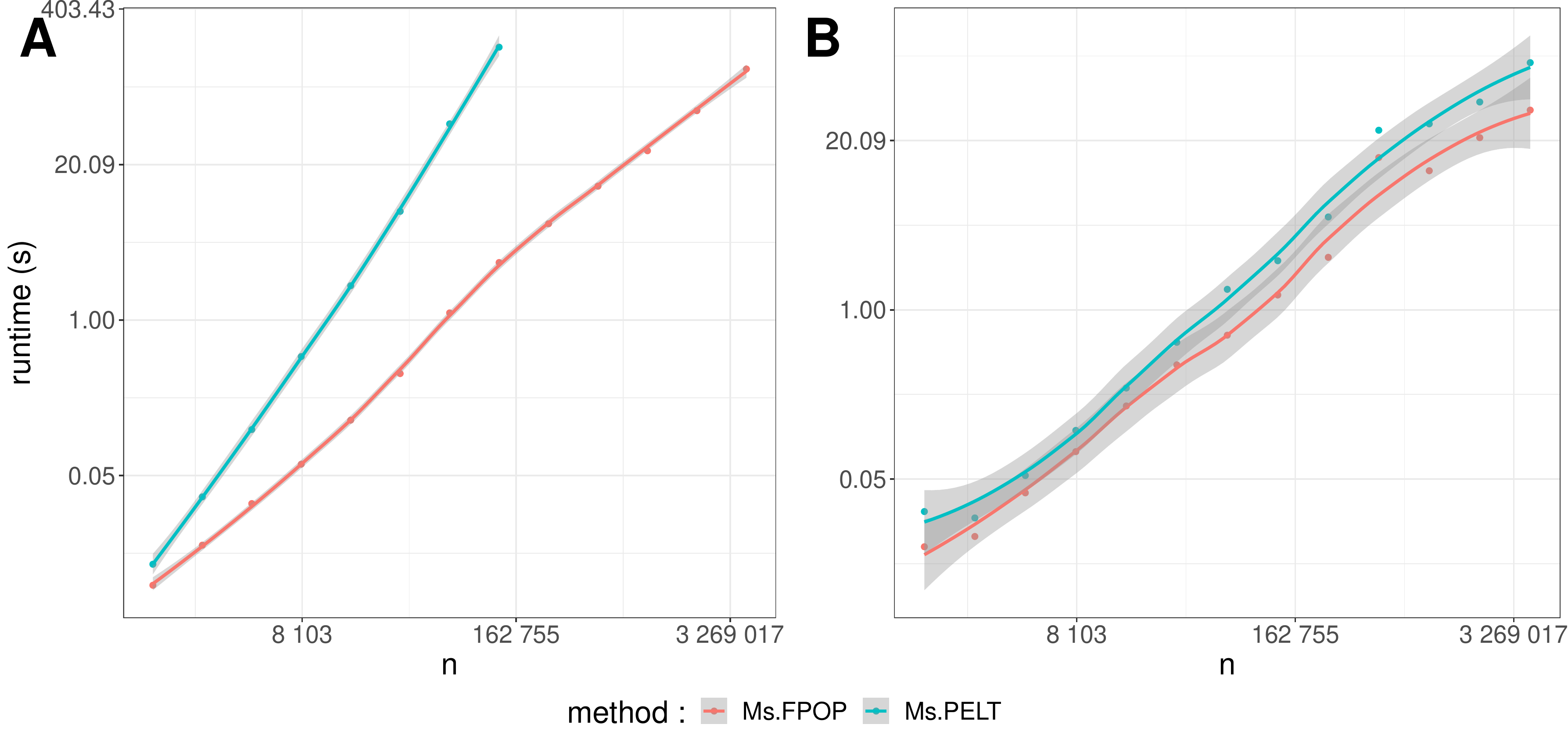}
\vspace{-0.7cm}
\caption{\textbf{Runtimes of \PELT{} and \PhiFPOP{} as a function of $n$ to optimize the multiscale penalty of \cite{verzelen2020optimal} with $\beta=2.25$ and $q=9$ on an Intel Core i7-10810U CPU @ 1.10GHzx12 computer for signals without changes (A) or signal with a change of size 1 every thousand datapoints (B).}}
\label{fig:simple_runtime}
\end{figure}

\paragraph{Outline}
In the rest of the paper we will (1) introduce our notations, (2) review the key idea behind \FPOP, (3) explain how and under which conditions we extend \FPOP\ to multiscale penalty, (4) study the performance of \PhiFPOP\ relative to \FPOP\ for various signals and (5) conclude with a discussion.

\subsection{Multiple Changepoint Model}

In this section we describe our changepoint notations and the multiscale criteria we want to optimize.

\paragraph{Segmentations and set of segmentations} For any $n$ in $\mathbb{N}$ we write $1:n=\{1, \cdots, n\}$. For any integer $D \geq 0$ we define a segmentation with $D$ changes of $1:n$ as an ordered subset of $1:(n-1)$ of size $D, $ with $\tau_ {j}$ the location of the $j^{th}$ change for $j$ in  $1, \ldots, D$. 
It will be usefull to also consider the dummy indices $\tau_{0} = 0$ and $\tau_{D + 1} = n$. We call $\mathcal{M}^{D}_{1: n}$ the set of all such segmentations in $D$ changes and $\mathcal{M}_{1: n}$ the union of all these sets : $\bigcup_{0 \leq D \leq n-1} \mathcal{M}^{D}_{1: n}.$ For any segmentation $\tau$ in $\mathcal{M}_{1:n}$ we note $|\tau|$ the number of segments of $\tau$. In other words, if $\tau$ is in $\mathcal{M}^{D}_{1:n}$ then $|\tau| = D+1.$ We can enumerate the elements of $\mathcal{M}_{1:n}$ and we get :
\begin{displaymath}
| \mathcal{M}_{1:n}| = \sum_{D = 0}^{n-1} | \mathcal{M}^{D}_{1:n} | = \sum_{D = 0}^{n-1} \binom{n-1}{D} = 2^{n-1}
\end{displaymath} 

\paragraph{Multiscale penalized likelihood} Under the piecewise constant model \ref{eq:piecewisemodel} a classical method to estimate the position and the number of changes is to optimize a penalized likelihood criterion. 
It is common to use a penalty that is linear in the number of changepoints  \cite{Yao,PELT,FPOP} and optimization wise the goal is to compute:

\begin{align*}
    {\TAU}^{*}_n & = &\argmin_{\tau \in  \mathcal{M}_{1:n}}\ \left\{ \displaystyle\sum_{j=1}^{|\tau|}  \min_{\mu} \left( \displaystyle\sum_{i=\tau_{j-1}+1}^{\tau{j}} (y_{i} -\mu)^2 \right) + \alpha|\tau| \right\}, \\
    F_n & = & \min_{\tau \in  \mathcal{M}_{1:n}}\ \left\{ \displaystyle\sum_{j=1}^{|\tau|}  \min_{\mu} \left( \displaystyle\sum_{i=\tau_{j-1}+1}^{\tau{j}} (y_{i} -\mu)^2 \right) + \alpha|\tau| \right\},
    \end{align*}
\vspace{-1cm}\begin{equation}\label{mainl2pen}
\end{equation}
where $\alpha$ is a constant to be calibrated (\textit{e.g.} $\alpha = 2 \log(n)$). 

Here we consider a more general penalty that depends on the length of the segments: 
\begin{align*}
    {\TAU}^{*}_n & = &\argmin_{\tau\in  \mathcal{M}_{1:n}}\ \left\{ \displaystyle\sum_{j=1}^{|\tau|}  \min_{\mu} \left( \displaystyle\sum_{i=\tau_{j-1}+1}^{\tau{j}} (y_{i} -\mu)^2 -\beta g(\tau_{j}-\tau_{j-1})\right) + \alpha |\tau|\right\}, \\
   F_n & = & \min_{\tau\in  \mathcal{M}_{1:n}}\ \left\{ \displaystyle\sum_{j=1}^{|\tau|}  \min_{\mu} \left( \displaystyle\sum_{i=\tau_{j-1}+1}^{\tau{j}} (y_{i} -\mu)^2 -\beta g(\tau_{j}-\tau_{j-1})\right) + \alpha |\tau|\right\},
\end{align*}
\vspace{-1cm}\begin{equation}\label{mainfpoppsd}
\end{equation}
where $g$ is a function satistfying assumption A\ref{assumpt1} described in the next paragraph, and $\alpha$ and $\beta$ are constants to be calibrated. 
We recover the multiscale criteria proposed in \cite{verzelen2020optimal} taking $g=\log$, $\alpha = \gamma + \beta g(n)$, and $\gamma $ a constant that remains to be chosen.
We recover the classical penalty of \cite{Yao} taking $g=0$, $\alpha =  2\log(n).$

\begin{assumption}\label{assumpt1} 
 $h(t, s, s') = g (t-s') - g(t-s)$ is a non-decreasing function in $t$, and $ \lim_{t \to \infty} h(t, s, s') = 0 $ therefore $h(t, s, s') \leq 0$.
\end{assumption}
This assumption will be useful later to bound the difference between the cost of two changes $s$ and $s'$. Intuitively, assumption A\ref{assumpt1} states that $g$ favors older changes but that asymptotically (large enough $t$ relative to $s$ and $s'$) this advantage for older changes vanishes. Importantly, this assumption is true for the multiscale penalty proposed in \cite{verzelen2020optimal} as $\beta > 0$ and $g (t-s') - g(t-s)= \log(1 - (s'-s)/(t-s))$ is increasing with $t$.

\subsection{Optimization with Dynamic Programming}

In this section we explain how one can optimize equation \eqref{mainfpoppsd} using dynamic programming ideas with (i) inequality based pruning and (ii) functional pruning.

\paragraph{Dynamic programming with inequality based pruning} The penalised cost of a segmentation $\tau$ inside the $\arg\min$ of equation \eqref{mainfpoppsd} can be written as a sum over all segments of $\tau$ :
\begin{displaymath}
\sum_{j=1}^{|\tau|}  \min_{\mu} \left( \displaystyle\sum_{i=\tau_{j-1}+1}^{\tau{j}} (y_{i} -\mu)^2 -\beta g(\tau_{j}-\tau_{j-1}) + \alpha \right),
\end{displaymath}
therefore the optimisation can be done iteratively using the Optimal Partionning (\OP) algorithm proposed in \cite{OP} using dynamical programming ideas developped in \cite{law} and \cite {bell}. It is possible to speed calculations using the \PELT\  algorithm \cite{PELT} because equation (4) of \cite{PELT} is true at least for constant $K = -\beta (\max_{1 \leq \ell \leq n}\{g(\ell)\} - 2 \min_{1 \leq \ell \leq n} \{ g(\ell) \})$ (see \ref{app:phi_pelt}). 
If $g$ is concave (as in the penalty \eqref{eq:criteria_verzelenetal} proposed in \cite{verzelen2020optimal}), $K$ can be chosen much closer to zero : $K = - \beta (g(2) - 2g(1))$ (see \ref{app:phi_pelt}), or adaptively to the last segment length $\ell$ : $\ K_{\ell} = - \beta (g(\ell)+g(1)-g(\ell+1))$ (see \ref{phi_pelt_2}). Our implementation of \PELT{}  optimizing (\ref{mainfpoppsd}) with $g=\log$ and $K_{\ell}=-\beta\log(\frac{1}{\ell}+1)$ is called \PhiPELT. Note that $K_{\ell}\leq -\beta\log(2)$.

As shown in the Figure \ref{fig:simple_runtime}, if the number of real changepoints is not linear in $n$, for $g = \log$, and a positive $\beta$, \PhiPELT\ is quadratic. This makes the analysis of large profiles with  $10^5$ or $10^6$ datapoints long and unpractical (\textit{e.g.} $>100$ seconds for a profile with $10^5$ datapoints and 1 changepoint, $>1$ hour for a profile with $10^6$ datapoints and 1 changepoint \footnote{ Runtimes observed on an Intel Core i7-10810U CPU @ 1.10GHzx12 computer.}).

\paragraph{Dynamic programming with functional pruning} In the rest of the paper, we present a functional pruning algorithm (called \PhiFPOP), in the sense of the \PDPA\ \cite{PDPA} or \FPOP\ \cite{FPOP}, to solve (\ref{mainfpoppsd}), making it possible to optimize (\ref{mainfpoppsd}) in a matter of seconds even for $n=10^6$ As the cost of equation (\ref{mainfpoppsd}) is not point-additive, condition C1 of \cite{FPOP} is not true, and maintaining the set of means for which a change is optimal is more complex. Our key idea is to maintain a slightly larger set that is easier to update.

\section{Functional Pruning}

\subsection{Functional Pruning Optimal Partioning (\FPOP{})}
To better explain \PhiFPOP{} we first review some of the key elements of \FPOP{} to optimize equation \eqref{mainl2pen}.
\FPOP{} introduces for every change $s$ its best cost as function of the last parameter $\mu$ at time $t$, $\widetilde{f}_{t, s} (\mu)$. 
Formally this is:
\begin{equation}\label{fpop_f}
    \widetilde{f}_{t,s}(\mu) = F_{s} + \displaystyle\sum_{i=s+1}^{t} (y_{i}-\mu)^{2} + \alpha, \text{  }\text{  } \text{ with } \text{  }\text{  } \widetilde{f}_{t,t}(\mu) = F_{t}+\alpha \text{  }\text{  }\text{ and } \text{  }\text{  } F_{0}=-\alpha.
\end{equation}
$\widetilde{f}_{t,s}(\mu)$ is a second degree polynomial in $\mu$ defined by three coefficients : $a_2\mu^2 + a_1\mu + a_0$ with 
$a_2  = t-s,$ $a_1  =-2\sum_{i=s+1}^{t}y_i$ and
$a_0  = F_{s}+\alpha+\displaystyle\sum_{i=s+1}^{t} y_{i}^2$.
The update of these coefficients is straightforward using the following formula: 
\begin{equation}
    \widetilde{f}_{t,s}(\mu) = \widetilde{f}_{t-1,s}(\mu) + (y_t-\mu)^2.
\end{equation}
    
At each time step $t$, \FPOP{} updates the minimum of all $\widetilde{f}_{t, s} (\mu) $, denoted $\widetilde{F}_{t}(\mu) =  \min_{s\leq t}\ \left\{ \widetilde{f}_{t,s}(\mu) \right\}.$ The key idea behind \FPOP{} is that to compute and update $\widetilde{F}_{t}(\mu)$ one only need to consider changes $s$ with a none empty ``living-set'' : $\mathcal{F}_{t}= \{ s \leq t | Z_{t, s}^{*} \neq \emptyset\}$ where the ``living-set'' of change $s$ is $Z_{t, s}^{*} = \{ \mu |  \widetilde{f}_{t, s}(\mu) = \widetilde{F}_{t}(\mu)\}$. Given those definitions we have $\widetilde{F}_{t}(\mu) =  \min_{s \in \mathcal{F}_{t}}\ \left\{ \widetilde{f}_{t,s}(\mu) \right\}$. In other words, $s$ is pruned as soon as its ``living-set'' is empty, which is justified because
\begin{equation}
    \label{14}
    Z_{t,s}^{*} \supset Z_{t+1,s}^{*}. \text{  }\text{  } \text{  }\text{and } \text{  }\text{  } \text{  } Z_{t,s}^{*} = \emptyset \implies Z_{t+1,s}^{*}=\emptyset\ .
\end{equation} 
Note that we can then retrieve $F_t$ by minimizing $\widetilde{F}_{t}(\mu)$ on $\mu$.

\subsection{\PhiFPOP\ : functional Pruning for a Multiscale Penalty}\label{new}

\PhiFPOP{} optimizes equation \eqref{mainfpoppsd}. 
As for \FPOP{} we introduce for every change $s$ its best cost as a function of the last parameter $\mu$ at time $t$, $\widetilde{f}_{t, s} (\mu)$. Formally this is :
\begin{equation}\label{new_f_cost}
    \widetilde{f}_{t,s}(\mu) = F_{s} + \displaystyle\sum_{i=s+1}^{t} (y_{i}-\mu)^{2} + \alpha - \beta g\left(t-s\right),
\end{equation}
with $\widetilde{f}_{t,t}(\mu) = F_{t}+\alpha$ and $F_{0}=-\alpha$. As in \FPOP{}, $\widetilde{f}_{t,s}(\mu)$ can be stored as a second degree polynomial in $\mu$. The update is also straightforward using the following formula:

\begin{equation}\label{mise_a_jour_cost_f_fpoppsd}
    \widetilde{f}_{t,s}(\mu) = \widetilde{f}_{t-1,s}(\mu) + (y_{t}-\mu)^{2} + \beta g\left(t-1-s\right) - \beta g\left(t-s\right)
\end{equation}

Analogously to \FPOP{} we can calculate $F_t$ by minimizing $\widetilde{f}_{t, s}$ both on $\mu$ and $s$. The main difference with \FPOP{} is that the rule (\ref{14}) is no longer true for \PhiFPOP{} because $\widetilde{f}_{t, s}(\mu) - \widetilde{f}_{t, s'}(\mu)$ depends on $t$:
\begin{equation} \label{etude_de_fct}
    \widetilde{f}_{t,s}(\mu) - \widetilde{f}_{t,s^{'}}(\mu) =  F_{s}-F_{s^{'}} + \displaystyle\sum_{i=s+1}^{s{'}} (y_{i}-\mu)^{2}
    +\beta(\underbrace{g(t-s^{'}) - g(t-s))}_{\substack{\text{\textit{a function varying}}\\ \text{\textit{with t, s et s'}}}}).
\end{equation}
Because of that, in the course of the algorithm we need to re-evaluate the set on which the candidate change $s$ is better than $s'$ at various $t$, $I_{t,s,s^{'}}$ with $s < s^{'}$: 
\begin{equation}
    I_{t,s,s^{'}} = \{ \mu\ |\ \widetilde{f}_{t,s}(\mu)\ \leq\ \widetilde{f}_{t,s^{'}}(\mu)\}.
\end{equation}

For arbitrary functions $g$ the set $I_{t,s,s^{'}}$ may vary drastically from one $t$ to the next. Using assumption A\ref{assumpt1} we can control those variations. 

\subsubsection{Update of The Candidate Changes Living Set ($ Z_{t, s}$)}
Rather than evaluating the exact living set $Z_{t, s}^{*}$ of all changes, we are seeking to update a slightly larger set, $Z_{t, s}$, including $Z_{t, s}^{*}$ and such that if $Z_{t, s}$ is empty we can guarantee that $Z_{t+h,s}^{*}$ is also empty for all $h>0$. The possibility of defining such a $Z_{t,s}$ depends on the property of the function $g$.

Assume A\ref{assumpt1} we propose to update $Z_{t+1,s}$ as follow:
\begin{equation} \label{update_croissant}
    Z_{t+1,s} = Z_{t,s}\ \cap \overbrace{(\ \bigcap_{s' \in \Future_{t,s}} I_{t+1,s,s{'}})}^{\text{\textit{comparison with future changes}}} \backslash \overbrace{(\ \bigcup_{s'' \in \Past_{s}} I_{\infty,s{''},s})}^{\text{\textit{comparison with past changes}}}\ ,
\end{equation}
where $\Future_{t, s}$ is any subset of $\{s+1, ..., t\}$,
$\Past_{s}$ is any subset of $\{1, ..., s-1\}$, and
$I_{\infty,s,s^{'}}$ correspond to $I_{t,s,s^{'}}$ when $t \to \infty$ (which is properly define under assumption A1). 

\paragraph{Pruning} Based on update \eqref{update_croissant} it should be clear that if $Z_{t,s}$ is empty so are all $Z_{t+h,s}, $ for $h>0.$  In the next lemma we show that $Z_{t,s}$ includes $Z^*_{t,s}.$ Therefore we further have that if $Z_{t,s}$ is empty so are all $Z^*_{t+h,s}, $ and change $s$ can be pruned. 

\begin{lemma}
Taking $Z_{s,s} = ]min_i y_i, \max_i y_i[$, updating $Z_{t+1,s}$ using equation \eqref{update_croissant} and assuming A\ref{assumpt1} we have
\begin{equation}\label{eq:lemmamain}
     Z_{t,s}^{*} \subset Z_{t,s}\ ,
\end{equation}
and for an integer $h > 0$
\begin{equation}\label{eq:lemmacorro}
Z_{t+h,s}^{*} \subset Z_{t+1,s}\ .
\end{equation}
\end{lemma}

\begin{proof}
For any $t$, we will prove by induction that for any $t'$ in $\{s, \cdots, t\}$ we have $Z_{t,s}^{*} \subset Z_{t',s}.$

For $t'=s$ and for any $t$ larger or equal to $s$ we have (by definition of $Z_{s, s}$) that $Z_{t,s}^{*} \subset Z_{t',s} = Z_{s, s}.$

Now assume that for $t' < t$ we have $Z_{t,s}^{*} \subset Z_{t',s}.$
As $h$ is non-decreasing for any $t'+1 \leq t$ we have the following two inclusions :
\begin{align}
I_{t, s, s^{'}}  & \subset  I_{t'+1, s, s^{'}}. \label{inclus2} \\
I_{\infty, s, s^{'}}  & \subset  I_{t'+1, s, s^{'}} \label{inclus1}  
\end{align} 

Therefore for $t'< t$ we have
\begin{eqnarray*} \label{h_croissant}
    Z^{*}_{t,s} & = & \qquad \quad (\ \bigcap_{s < s^{'} \leq t} I_{t,s,s^{'}})   \backslash (\ \bigcup_{s^{''} < s} I_{t,s{''},s}) \qquad \quad \quad \text{ by definition of } Z^*_{t,s} \\
    Z^{*}_{t,s} & \subset & Z_{t',s}\ \cap (\ \bigcap_{s < s^{'} \leq t} I_{t,s,s^{'}})   \backslash (\ \bigcup_{s^{''} < s} I_{t,s{''},s}) \qquad \quad \quad \text{by induction} \\
    & \subset &  Z_{t',s}\ \cap (\ \bigcap_{s \ < s^{'} \leq \ t} I_{t'+1,s,s^{'}}) \backslash (\ \bigcup_{s^{''} < \ s} I_{\infty,s^{''},s}) \quad \text{using equation \eqref{inclus2} and\eqref{inclus1}} \\
    & \subset &  Z_{t',s}\ \cap (\ \bigcap_{s^{'} \in \ \Future_{t', s}} I_{t^{'}+1,s,s^{'}}) \backslash (\ \bigcup_{s^{''} \in \ \Past_{s}} I_{\infty,s^{''},s}) \quad \text{by definition of } \Future_{t', s} \text{ and } \Past_{s}.\\
\end{eqnarray*}
Using equation \eqref{update_croissant} we thus get that $Z^{*}_{t,s} \subset Z_{t'+1, s}, $ proving the induction.

To recover equation \eqref{eq:lemmacorro} we notice from update \eqref{update_croissant} that $Z_{t+1, s} \subset Z_{t, s}$ and apply equation \eqref{eq:lemmamain}.
\end{proof}

\subsubsection{\PhiFPOP\ Algorithm, Choice of $\Future_{t, s}$ and $\Past_{s}$}\label{sec:subsampling}

The update rule (\ref{update_croissant}) suggest that for each candidate change $s$ we should compare it future change $s'$ in $\Future_{t,s}$, and past change $s''$ in $\Past_{s}$. For past candidate changes $s^{''}$ this comparison can be done once and for all considering that $t$ goes to infinity ($I_{\infty,s^{''},s}$). For future candidate changes $s^{'}$, on the contrary, it might be usefull to update the interval $I_{t,s,s^{'}}$.
Performing at each time step, for each $s$, a comparison with all s' is time consuming. Intuitively, the complexity of each time step is in $\mathcal{O}(\text{number of candidate changes}^2).$ Ideally, for each $s$, one would like to make the minimum number of comparisons that would result in its pruning. In the  Algorithm \ref{algo1} we consider a generic \textit{sampling} function of $s'$ that returns $\Future_{t,s}$ (see the Sampling Strategies paragraph in section \ref{sec:Rcpp}). \\[0.1cm]

{\footnotesize
\begin{algorithm}[H]
\SetAlgoLined
\KwIn{$Y=(y_{1}, ..., y_{n})$, $\ \alpha$,$\ \beta$, $\ g=log(.)$}
\KwOut{set of last best changes $cp_{n}$}
$n \gets |Y|$\;
$F_{0} \gets -\alpha$\;
$cp_{0} \gets \emptyset$\;
$R_{1} \gets \{0\}$\;
$D \gets [min(Y),\ max(Y)]$\;
$Z_{0,0} \gets D$\;
$\widetilde{f}_{0,0} \gets F_{0} + \alpha\ (=0)$\;
\For{$t \gets 1,...,n$}{
    \For{$s \in R_{t}$}{
    $\widetilde{f}_{t,s}(\mu) \gets \widetilde{f}_{t,s}(\mu) + (y_{t}-\mu)^{2}+\beta \times g(t-1-s) - \beta \times g(t-s)$\;
    }
    $F_{t} \gets \min_{s \in R_{t}}(\mintheta_{\mu \in Z_{t,s}}(\widetilde{f}_{t,s}(\mu)))$\;
    $s_{t} \gets \argmin_{s \in R_{t}}(\min_{\mu \in Z_{t,s}}(\widetilde{f}_{t,s}(\mu)))$\;
    $cp_{t} \gets (cp_{s_{t}}, s_{t})$\;
    $\widetilde{f}_{t,t} \gets F_{t} + \alpha$\;
    $Z_{t,t} \gets D$\;
    \For{$s \in R_{t}$}{
    $Z_{t,t} \gets Z_{t,t}\ \backslash\ I_{\infty, s, t}$\;
    $\Future_{t,s} \gets sample(\{s^{'} \in \{R_{t} \cup \{t\}\}: s^{'} > s\})$\; 
    $Z_{t, s} \gets Z_{t,s} \cap (\bigcap_{s^{'} \in\ \Future_{t,s}} I_{t,s,s^{'}})$\;
    }
    $R_{t+1} \gets \{s \in \{R_{t} \cup \{t\}\}: Z_{t,s} \neq \emptyset\}$\;
}
\caption{\PhiFPOP{} \label{algo1}}
\end{algorithm}}

\section{Rcpp Implementation of \PhiFPOP\ Algorithm}\label{sec:Rcpp}

\paragraph{\PhiFPOP\ R package} The dynamic programming and functional pruning procedures describe in the  algorithm \ref{algo1} are implemented in \textit{C++}. The input and output operations are interfaced with the R programming language thanks to \textit{Rcpp} R package. The main function \verb|MsFPOP()| takes as input the sequence of obervations, a vector of weights for these obervations, the parameters $\beta$ and $\alpha$ of the multiscale penalty. The function returns the set of optimal changepoints in the sense of (\ref{mainfpoppsd}). Analogously, we implemented a version of the \PELT\ algorithm, \verb|MsPELT()|, that optimizes (\ref{mainfpoppsd}). 

\paragraph{Sampling Strategies} To recover $\Future_{t,s}$ we consider either an exhaustive sampling of all future changes $s'>s$ in $R_t$ or a uniform random subsampling of them without replacement. The main function parameter \verb|size| can be set by the user to specify for each $s$ the number of sampled $s'$. In the appendix we compare the runtime of different sampling strategies (see \ref{supp_speed_benchmark}).\\[0.1cm]

\section{Simultation Study}

\subsection{Calibration of Constants $\gamma$ and $\beta$ from The Multiscale Penalty}\label{calibration}

Paper \cite{verzelen2020optimal} does not recommend values for $\gamma$ and $\beta$ in their penalty \eqref{eq:criteria_verzelenetal}. As explained in detail below, we calibrated those values to control the percentage of falsely detecting at least one change in profiles simulated without any actual change. 

\paragraph{No change simulation} We repeatedly simulate \textit{iid} Gaussian signals of mean 0, variance 1 and varying lengths $n$ ($n \in \{10^2, 10^3, 10^4, 10^5, 2.5\times10^5\}$). On these profiles we run \PhiFPOP\ for different $\gamma$ values (ten $\gamma$ values evenly spaced on the interval $[1,20]$) and different $\beta$ values ($\beta \in \{2, 2.25, 2.5, 2.75, 3\})$.

\paragraph{Percentage of false detection} We denote $R_{>0}$ as the proportion of replicates for which \PhiFPOP\ returns at least 1 changepoint. These changepoints are false positives. Our goal is to find a combination of $\beta$ and $\gamma$ such that 
\begin{equation} \label{signi}
    R_{>0} < 0.05\ (\text{significance level}) \quad.
\end{equation}\paragraph{Empirical Results} In Figure \ref{fig:cal_phifpop_main} we observe that, by setting $\beta = 2.25$, a conservative range of $\gamma$ satisfying inequality (\ref{signi}) can be reached for $\gamma \in [7.5,10]$. Note that this interval satisfy inequality (\ref{signi}) for all tested $n$ and $\beta$ (see \ref{cal_phifpop_main_supp}). \vspace{0.3cm}

\noindent Based on these results, in the following simulations we set $\gamma = 9$ and $\beta = 2.25$\footnote{This is equivalent to setting $L = 1.125$ and $q = 8$ in equations (31) and (32) of \cite{verzelen2020optimal}} for all methods optimizing (\ref{mainfpoppsd}) (\PhiFPOP, \PhiPELT). We set $\alpha=2\sigma^2\log(n)$ for all methods optimizing (\ref{mainl2pen}) (\FPOP, \PELT).

\begin{figure}[H]
\centering
\includegraphics[scale=0.43]{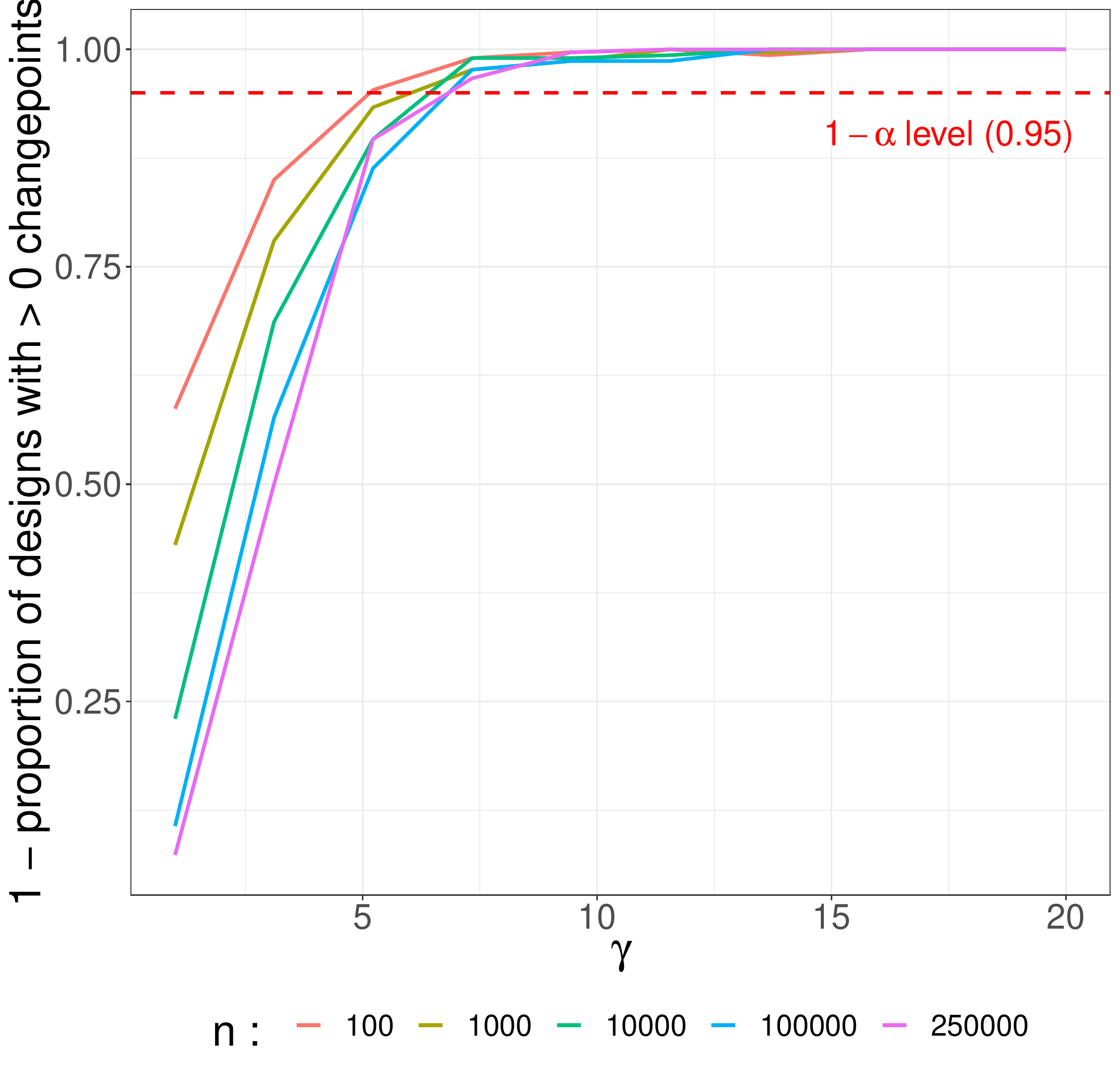}
\caption{\textbf{Proportion of stationary Gaussian signal replicates on which \PhiFPOP\ returns at least 1 changepoint ($R_{>0}$).} $R_{>0}$ is computed for a series of $\gamma$ and profile lengths (see \textit{Design of Simulations)}. In these simulations we set $\beta = 2.25$. Results for other $\beta$ values are availables in \ref{cal_phifpop_main_supp}.}
\label{fig:cal_phifpop_main}
\end{figure}

\subsection{Evalutation of \PhiFPOP: Speed Benchmark} \label{speed_bench}

\paragraph{Design of Simulations} We repeatedly simulate \textit{iid} Gaussian signals with $10^5$ datapoints. The profiles are affected by one or more changepoints in their mean ($D \in \{1, 5, 10, 15, 20, 25, 30, 45, 50, 100, 150, 200, 250, 300, 350, 400, 450,$\\$500, 550, 600, 650, 700, 750, 800, 850, 900, 950, 1000\}$). The mean of segments alternates between 0 and 1, starting with 0. The variance of each segment is fixed at 1. On these profiles we run two methods optimizing the penalized likelihood defines in (\ref{mainl2pen}): \PELT\ \cite{PELT} and \FPOP\ \cite{FPOP}, as well as methods optimizing the multiscale penalized likelihood defines in (\ref{mainfpoppsd}): \PhiPELT{} and \PhiFPOP{}. For \PhiFPOP{}, after comparisons with other sampling strategies (see \ref{supp_speed_benchmark}), we choose to randomly sample 1 future candidate change.

\paragraph{Metric} For each replicate we time in seconds the compared methods.

\paragraph{Empirical Results} In Figure \ref{fig:runtime} we firstly observe  that for both criteria (multiscale penalized likelihood and penalized likelihood), functional pruning methods are always faster than inequality based pruning ones. Indeed, \PhiFPOP{} and \FPOP{} are always faster than \PhiPELT{} and \PELT, respectively. The smaller $D$, the larger the time difference between functional pruning methods and inequality based pruning ones. For $D=1$, \PhiFPOP{} runs in 2.4 seconds in average and is about 50 times faster than \PhiPELT{} (121.3 seconds in average). For $D=1000$, \PhiFPOP{} runs in 0.7 second in average and is about 1.3 times faster than \PhiPELT{} (0.9 second in average). Marginally to $D$, \FPOP{} runs always under 0.05 seconds. Similar trends can be observed on \textit{iid} Gaussian signals with $10^6$ datapoints (see Figure \ref{fig:speed1M}).

\begin{figure}[H]
\centering
\includegraphics[scale=0.43]{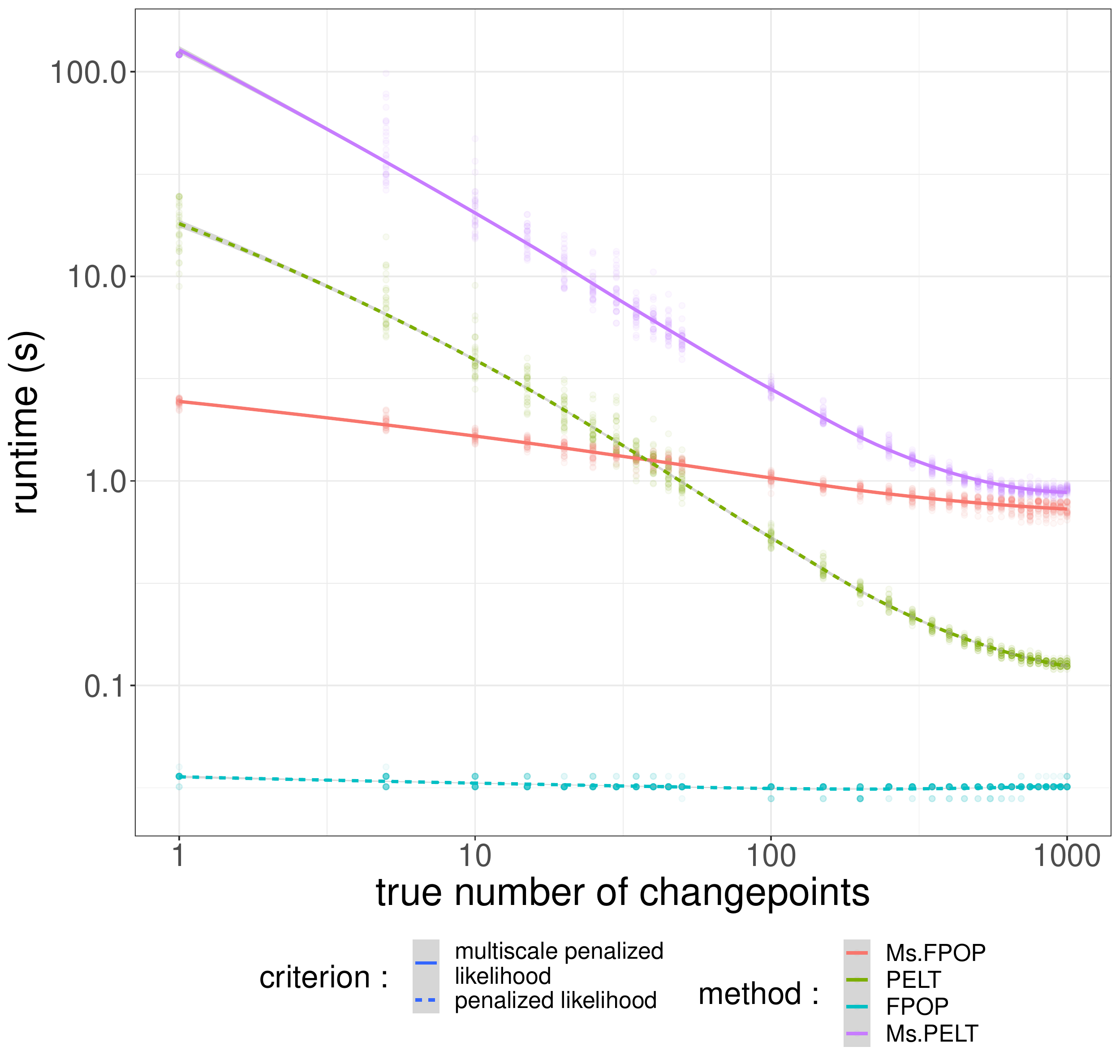}
\caption{\textbf{Runtimes as a function of the true number of changepoints.} We timed \PELT, \PhiPELT{}, \FPOP, \PhiFPOP{} on profiles of length $n=10^5$ with varying number of true changepoints $D$ (see \textit{Design of Simulations}) on an Intel Core i7-10810U CPU @ 1.10GHzx12 computer. The comparison between sampling strategies of future candidate changes implemented in \PhiFPOP\ is available in Figure \ref{fig:subsamplingspeed}. The comparison of \PELT, \FPOP, \PhiFPOP{} on profiles of length $n=10^6$ is available in Figure \ref{fig:speed1M}.}
\label{fig:runtime}
\end{figure}

\subsection{Evalutation of \PhiFPOP{} relative to \FPOP: Accuracy Benchmarks}

In this section we seek to illustrate using minimalist simulations the performances of the multiscale criteria proposed in \cite{verzelen2020optimal} and implemented in \PhiFPOP{} relative to the BIC criteria proposed in \cite{Yao} and implemented in \FPOP{}.

\subsubsection{Hat Simulations} 
\label{text:hat_simu}
\paragraph{Design of Simulations} We repeatedly simulate \textit{iid} Gaussian signals of varying size $n \in \{10^3, 10^4, 10^5\}$. Each signal is affected by 2 changepoints. The second changepoint ($\tau_2$) is fixed at position $\lfloor\frac{2n}{3}\rfloor$ while we vary the position of the first changepoint ($\tau_1$) (see Figure \ref{fig:hat_fig}.A). $\tau_1$ takes a series of 30 positive integers evenly spaced on the $\log$ scale on the interval $[1, \lfloor\frac{n}{3}\rfloor]$. We also look at the symmetry of this series builds around $\lfloor\frac{n}{3}\rfloor$ (i.e. $\lfloor\frac{2n}{3}\rfloor-\tau_1$, see dotted line in Figure \ref{fig:supphat_simu}). Note that for $\tau_1 = \lfloor\frac{n}{3}\rfloor$ the segmentation is balanced. The means of the three resulting segments are set to $\mu_1 = 0$ , $\mu_2=\sqrt{\frac{100}{n}}$ and $\mu_3 = 0$. We run both \PhiFPOP{} and \FPOP{} on these profiles. \PhiFPOP{} incorporates a multiscale penalty, while \FPOP{} assigns equal weight to all segment sizes and serves as a reference point for comparison with \PhiFPOP. We anticipate that the multiscale penalty in \PhiFPOP{} will lead to more accurate segmentations of profiles with well-spread changepoints compared to \FPOP. Additionally, as the size of the data ($n$) increases, we expect \PhiFPOP{} to get similar performance or outperform \FPOP{} in terms of accuracy for all segment sizes.

\paragraph{Metric} We denote $R_{2}$ the proportion of replicates for which a method returns exactly two changepoints. We also denote $\Delta_{R_2}$, the $\log_2$-ratio between $R_{2}$ of \PhiFPOP{} and \FPOP.

\paragraph{Empirical Results} 
In Figure \ref{fig:hat_fig}.B and \ref{fig:supphat_simu} we observe that with both \PhiFPOP{} and \FPOP, $R_{2}$ increases when $\tau_1$ tends towards $\lfloor\frac{n}{3}\rfloor$ (balanced segmentation). Note that the maximum is reached before $\tau_1 = \lfloor\frac{n}{3}\rfloor$.

Furthermore, in agreement with our expectations, in Figure \ref{fig:hat_fig}.B we observe that $\Delta_{R_2}$ increases when $\tau_1$ tends towards $\lfloor\frac{n}{3}\rfloor$. When $n$ increases, the differences observed on small segments in favor of \FPOP{} ($\Delta_{R_2}<0$) disappear ($\Delta_{R_2}\to 0$) and the differences on other segments in favor of \PhiFPOP{} ($\Delta_{R_2}>0$) are accentuated.

\begin{figure}[H]
\centering\includegraphics[scale=0.5]{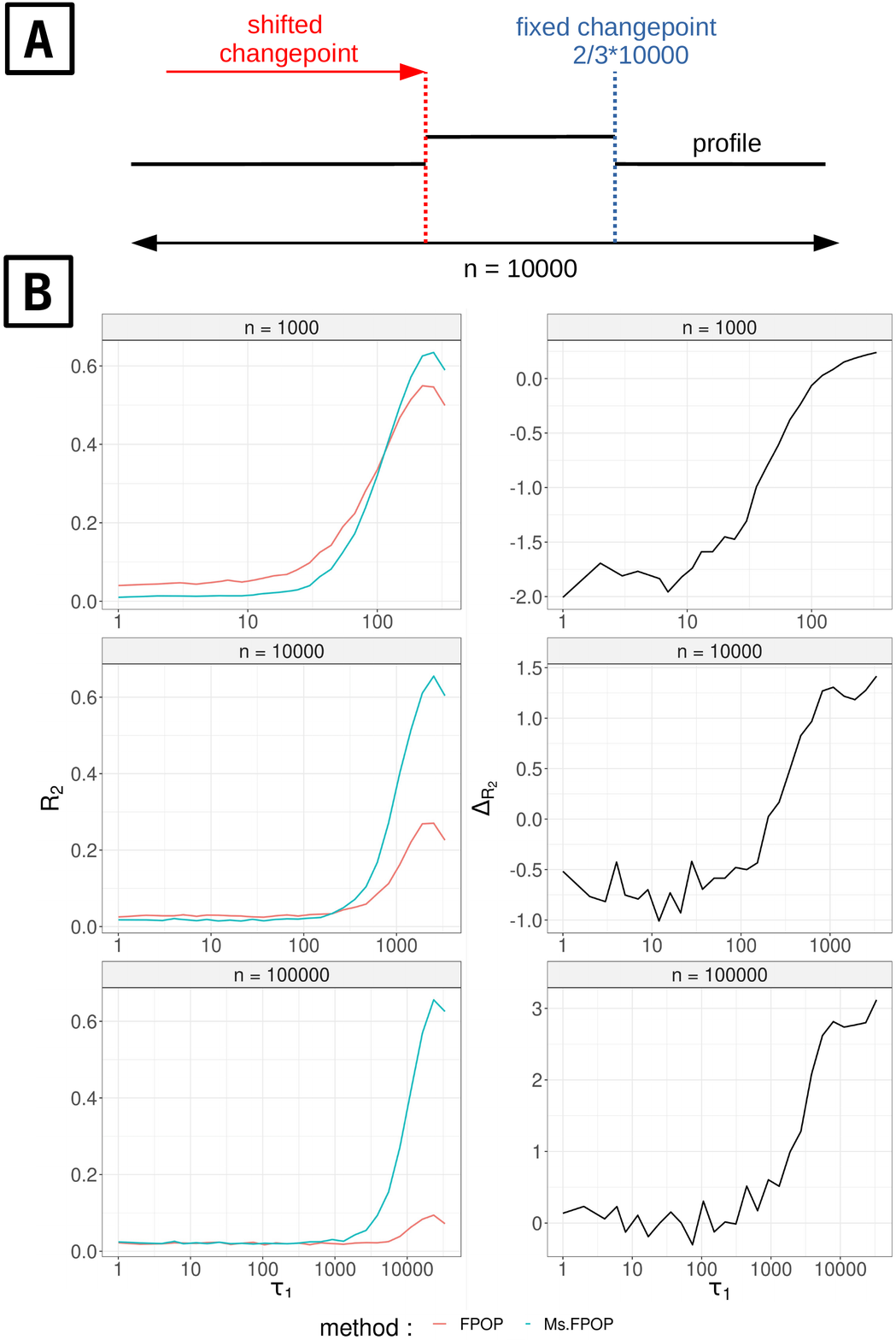}

\caption{\textbf{\PhiFPOP\ increases the probability of finding well spread changepoints on \textit{hat simulations}.} - (A) denoised profile with 2 changepoints. The second changepoint is fixed at $\lfloor\frac{2n}{3}\rfloor$ while the first one ($\tau_1$) varies on the interval $[1,\lfloor\frac{n}{3}\rfloor$]. The means of the three resulting segments are set to $\mu_1 = 0$ , $\mu_2=\sqrt{\frac{100}{n}}$ and $\mu_3 = 0$ which gives the profile a hat-like appearance. An \text{iid} Gaussian noise of mean 0 and variance one is then added (see \textit{Design of simulations}). - (B) The proportion of replicates for which \PhiFPOP\ and \FPOP\ return 2 changepoints ($R_2$) as well as the $\log_2$-ratio of the two estimations ($\Delta_{R_2}$) are computed for varying $\tau_1$ and $n$.}
\label{fig:hat_fig}
\end{figure}

\subsubsection{Extended Range of Simulation Scenarios}
\label{extended_sim}
\paragraph{Design of Simulations} Following a protocol written by Fearnhead \textit{et al.} 2020, we simulate different scenarios of \textit{iid} Gaussian signals. Each scenario is defined by a combination of $D$, $n$, $\tau$, $\mu$. For each scenario we vary the variance $\sigma^2$ (see Supplementary Data of \cite{Fearnhead2020}). All the simulated profiles, with a variance one, can be seen in \ref{supp:other_simu_default}. Based on these initial scenarios we simulate another set of profiles in which profile lengths are multiplied so that each segments contain at least 300 datapoints. These new set of simulated profiles can be seen in \ref{supp:other_simu_min300}. For each scenario and tested $\sigma^2$ we simulate 300 replicates.

\paragraph{Metric} We denote $AE\%$, the average number of times a method is at least as good as other methods in terms of absolute difference between the true number of changes and the estimated number of changes ($\Delta_D$), mean squared error (MSE) or adjusted rand index (ARI). The closer to 100 (AE\%), the better the method. See Supplementary Data of \cite{Fearnhead2020} for a formal definition of this criterion. 

\paragraph{Empirical Results} On the simulation of \cite{Fearnhead2020} in which a large portion of the segments have a length under 100 the performance of \PhiFPOP\ are worse than \FPOP\ and \textbf{MOSUM} \cite{meier2021mosum} on almost all scenarios except \textit{Dt7} that do not contain any changepoint (see \ref{supp:other_simu_default}).

On the second set of profiles, using $\Delta_D$ as comparison criterion, we observe on Figure \ref{fig:AE_perc_K_min300} that \PhiFPOP\ get similar performance or is better than \FPOP\ and \textbf{MOSUM} in all scenarios marginaly to $\sigma^2$. The results are similar when we use MSE or ARI as a criterion of comparison (see \ref{supp:other_simu_min300}). 

\begin{figure}[h!]
    \centering\includegraphics[scale=0.33]{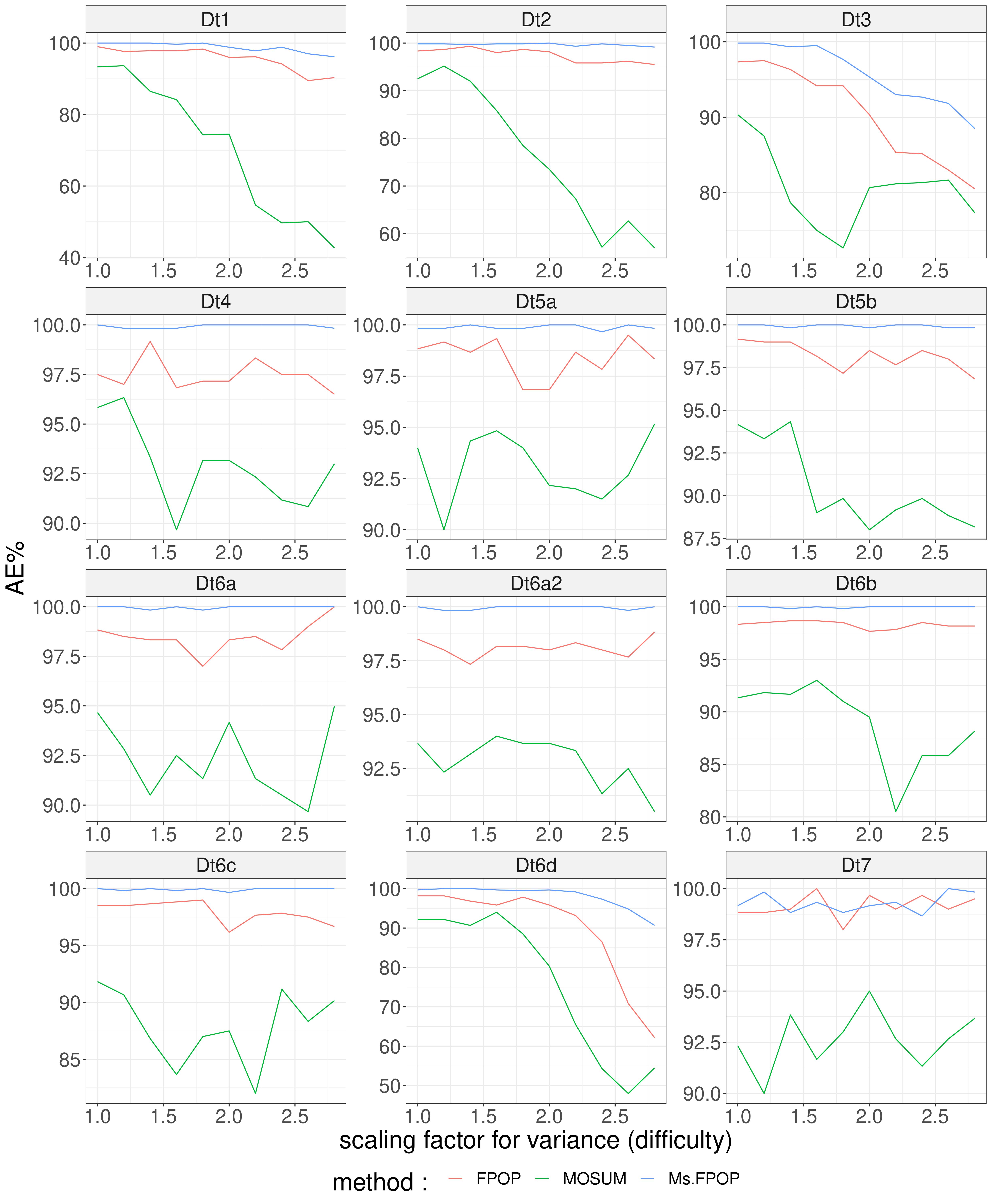}
    \vspace{-0.5cm}\caption{\textbf{AE\% as a function of the scaling factor for the variance (comparison criterion : $\Delta_D$).} The average number of times a method is at least as good as other methods in terms of $\Delta_D$ is computed for \FPOP, \PhiFPOP, and \textbf{MOSUM} on different scenarios of \textit{iid} Gaussian signals and varying $\sigma^2$. The smallest segment length is greater or equal to 300 (see \textit{Design of Simulations}). Each panel stands for the results on one scenario. Corresponding profiles can be viewed in \ref{supp:other_simu_min300}.}
    \label{fig:AE_perc_K_min300}
\end{figure}

\section{Discussion}

\paragraph{Extending Functional Pruning Techniques to the Multiscale Penalty} In section \ref{new} we have explained how to extend functional pruning techniques to the case of multiscale penalty. In Figures \ref{fig:simple_runtime} and \ref{fig:runtime} we have seen that for large signals ($n\geq10^5$) with few changepoints,  \PhiFPOP\ is an order of magnitude faster than \PhiPELT\ (which relies on inequality based pruning, see \ref{app:phi_pelt}). Even when the number of changepoints increased linearly with the size of the data, \PhiFPOP\ was still faster than \PhiPELT. 

The main update rule \ref{update_croissant} of our dynamic programming algorithm suggests to compare each candidate change $s$ with a set of future candidate changes $s'$. As we have seen in \ref{supp_speed_benchmark}, the strategy of randomly drawing one $s'$ according to a uniform distribution is the best strategy and allows us to tackle large signals. It is likely that uniform sampling is not optimal. The algorithm alternates between good draws (leading to a strong reduction of $Z_{t,s}$ or even the pruning of $s$) and bad draws (leading to a weak reduction $Z_{t,s}$). On average this is sufficient but improvements are possible. In particular the study of $h(t,s,s')=\log(\frac{t-s'}{t-s})$ (see Assumption A1), suggests disfavoring $s'$ that are too recent or that have been compared recently.

\paragraph{Calibration of $\gamma$ and $\beta$ from the Multiscale Penalty} The least-squares estimator with multiscale penalty proposed by \cite{verzelen2020optimal} involves two constants $\gamma$ and $\beta$ that still need to be investigated. Using signals simulated under the null hypothesis (no changepoint) we have seen that it is possible to find a pair of constants $\gamma = 9$ and $\beta = 2.25$ for which \PhiFPOP{} controls $R_{>0}$. Under this setting we have shown on \textit{hat} (see  section \ref{text:hat_simu}) and \textit{step} (see Figure \ref{fig:suppstep_simu}) simulations that \PhiFPOP\ is more powerful than \FPOP\ on segmentations with well-spread changepoints. This difference of power grows with $n$. For segmentation with small segments \FPOP\ is more powerful \PhiFPOP\ when $n$ is small ($\approx 10^3$), but for larger $n$ ($\geq 10^4$) this difference disappears.

We also tested \PhiFPOP{} on the benchmark proposed in \cite{Fearnhead2020}. The performances of \PhiFPOP{} are not so good on the original benchmark containing mostly small profiles with small segments but much better for an extended benchmark with larger profiles (see section \ref{extended_sim}). 

Without additional work on the calibration of the constants, we would thus recommend using \PhiFPOP for large profiles ($\geq 10^4$).

\paragraph{Unknown Variance} All our simulations have been done on signals with known variance, $\sigma^2$. However, in real-world situations, this may not always be the case. One approach is to estimate $\sigma^2$ and then plugging-in it in the problem, \textit{i.e} scaling the signal or the penalty by $\frac{1}{\sigma^2}$ or $\sigma^2$, respectively. A robust estimate of $\sigma^2$ can be obtained by calculating the variance of $\Delta_Y = Y_{i+1} - Y_i$ using either the median absolute deviation or the estimator suggested in \cite{hall1990asymptotically}. As an alternative, \cite{verzelen2020optimal} pointed out that one could calibrate the multiplicative constant $L$ of the penalized least-squares estimator using the slope heuristic \cite{Arlot2019}. Investigating the performances of these various approaches is outside the scope of this paper.

\section{Availability of Materials}
The scripts used to generate the figures are available in the following GitHub repository: \url{https://github.com/aLiehrmann/MsFPOP_paper}
. A reference implementation of the \PhiFPOP{} (and \PhiPELT) algorithm is available in the R package of the same name: \url{https://github.com/aLiehrmann/MsFPOP}.

\appendix

\section{\PELT\ for Multiscale Penalized Likelihood}\label{app:phi_pelt} 

\noindent Following the notation of the PELT paper \cite{PELT} the cost of a segment from $s+1$ to $s'$, $s+1:s'$ is defined as $\mathcal{C}_{s+1:s'} = \sum_{i=s+1}^{s'} (y_{i}-\bar{y}_{s+1:s'})^{2} - \beta g(s'-s).$
In what follow we consider three time points $s<s'<t$. Let $\ell = s'-s$ denote the length of the sequence of observations between time $s$ and $s'$ and $\ell' = t-s'$ denote the length of the sequence of observations between time $s'$ and $t$. 

The key condition to apply the PELT algorithm \cite{PELT} is that up to a constant $K$ adding a changepoints always reduce the cost, that is :
\begin{assumption}\label{eq:peltcondition}
 \begin{equation}
    \mathcal{C}_{s+1:s'} + \mathcal{C}_{s'+1:t} + K \leq \mathcal{C}_{s+1:t}
\end{equation}
\end{assumption}

The following lemma ensure that such $K$ exists for any $n$ and provide explicit values for $K$ in general and if $g$ is concave.

\begin{lemma}\label{lemma:peltbound}
(a) For any function $g$ from $\mathbb{R}$ to $\mathbb{R}$, $\beta \geq 0$, and any $n$, Assumption \ref{eq:peltcondition} is true at least for $K=2 \beta \min_{1 \leq \ell \leq n } \{g(\ell)\} -\beta \max_{1 \leq \ell \leq n } \{g(\ell)\}$.
(b) If $g$ is concave the condition is true for $K=-\beta g(2) + 2\beta g(1).$

\end{lemma}

\begin{proof}
We first note that 
\begin{displaymath}
m_n = \underset{1 \leq \ell < n}{\min} \left\{ \underset{\substack{1 \leq \ell' < n \\ \ell+\ell' \leq n}}{\min} \left\{ g(\ell)+g(\ell') - g(\ell+\ell') \right\} \right\}
\end{displaymath}
is well defined as the minimum of a finite set.
By definition of $m_n$ we thus have, for any $1 \leq s < s' < t \leq n$ and for any $K < \beta m_n$, that 
\begin{eqnarray*}
-\beta g(s'-s) - \beta g(t-s') + K & \leq & -\beta g(t-s)
\end{eqnarray*}
Combining this with 
\begin{eqnarray*}
\sum_{i=s+1}^{s'} (y_{i}-\bar{y}_{s+1:s'})^{2} + \sum_{i=s'+1}^{t} (y_{i}-\bar{y}_{s'+1:t})^{2} & \leq & \sum_{i=s+1}^{t} (y_{i}-\bar{y}_{s+1:t})^{2}.
\end{eqnarray*}
we recover that equation \eqref{eq:peltcondition} is true for any $K < \beta m_n.$

Now for any $\ell$, $\ell'$ in $\{1, \ldots, n\}^2$ such that $\ell + \ell' \leq n$ we have
\begin{displaymath}
2 \min_{1 \leq \ell \leq n } \{g(\ell)\} - \max_{1 \leq \ell \leq n } \{g(\ell)\} \leq g(\ell) + g(\ell') - g(\ell+\ell').
\end{displaymath}
Hence we get 
\begin{displaymath}
2 \min_{1 \leq \ell \leq n } \{g(\ell)\} - \max_{1 \leq \ell \leq n } \{g(\ell)\} \leq m_n,
\end{displaymath}
and we recover (a).

In case $g$ is concave using the technical lemma \ref{lemma:concavitydecreasing} two times we get :
\begin{equation}\label{eq:lspecificPELT}
\underset{\substack{1 \leq \ell' < n \\ \ell+\ell' \leq n}}{\min} \left\{ g(\ell)+g(\ell') - g(\ell+\ell') \right\} = g(\ell) + g(1) - g(\ell+1)
\end{equation}
and 
\begin{displaymath}
\underset{1 \leq \ell < n}{\min} \left\{ \underset{\substack{1 \leq \ell' < n \\ \ell+\ell' \leq n}}{\min} \left\{ g(\ell)+g(\ell') - g(\ell+\ell') \right\} \right\} = 2g(1) - g(2) 
\end{displaymath}

For example, if $g=\log$ we get $K= -\beta\log(2)$
\end{proof}

\begin{lemma}\label{lemma:concavitydecreasing}
If $g$ is concave then for any $\delta>0$, the function $h: x \rightarrow g(x+\delta) - g(x)$ is non increasing.
\end{lemma}
\begin{proof}
Consider any $\delta' > 0$. 
We have 
$x+\delta = (1-\alpha) x + \alpha(x+\delta+\delta')$ for $\alpha = \delta/(\delta+\delta')$ and similarly
$x+\delta' = (1-\alpha') x + \alpha'(x+\delta+\delta')$ with $\alpha' = \delta'/(\delta+\delta').$
Using concavity we have
\begin{eqnarray*}
g(x+\delta) & \geq &(1-\alpha) g(x) + \alpha g(x+\delta+\delta') \\
g(x+\delta') & \geq & (1-\alpha') g(x) + \alpha' g(x+\delta+\delta'). \\
\end{eqnarray*}
Suming these two lines and noting that $\alpha+\alpha'=1$ we get 
$g(x+\delta) - g(x) \geq g(x+\delta'+\delta) - g(x+\delta')$

\end{proof}

\section{Adaptative \PELT{} for Concave Multiscale Penalty }\label{phi_pelt_2} 

In the following lemma we show that for our multiscale penalty assuming the function $g$ is concave the constant $K$ in theorme 3.1 of \cite{PELT} can be chosen adaptively to the length of the last segment.

\begin{lemma}
If $g$ is concave and $\beta \geq 0.$
then if at time $s'$ we have,
$$ F_s + \sum_{i=s+1}^{s'} (y_{i}-\bar{y}_{s+1:s'})^{2} - \beta g(\ell) + K_{s'-s=\ell} \geq F_{s'}, $$ with $K_{\ell} = \beta( g(\ell) + g(1) - g(\ell+1))$
then for any time $t$ larger than $s'$ we have :
$$ F_s + \sum_{i=s+1}^{t} (y_{i}-\bar{y}_{s+1:t})^{2} - \beta g(\ell+\ell') \geq F_{s'} + \sum_{i=s'+1}^{t} (y_{i}-\bar{y}_{s'+1:t})^{2} - \beta g(\ell'), $$
and thus
for any time $t\geq s'$, a change at $s$ can never be optimal. Taking $g=\log$ we get $K_{\ell}= -\beta\log(\frac{1}{\ell}+1) \leq -\beta\log(2)$.

\end{lemma}
\begin{proof}
We follows the proof of Theorem 3.1 of \cite{PELT} using the fact that if $g$ is concave then equation \ref{eq:lspecificPELT} is true.
\end{proof}

\section{\PhiFPOP\ : Calibration of Constants $\gamma$ and $\beta$ from The Multiscale Penalty}
\label{cal_phifpop_main_supp}

The following plots were generated to calibrate the constants in the multiscale penalty of \cite{verzelen2020optimal}. They are generated as explained in section \ref{calibration}.

\begin{figure}[H]
    \centering\includegraphics[scale=0.43]{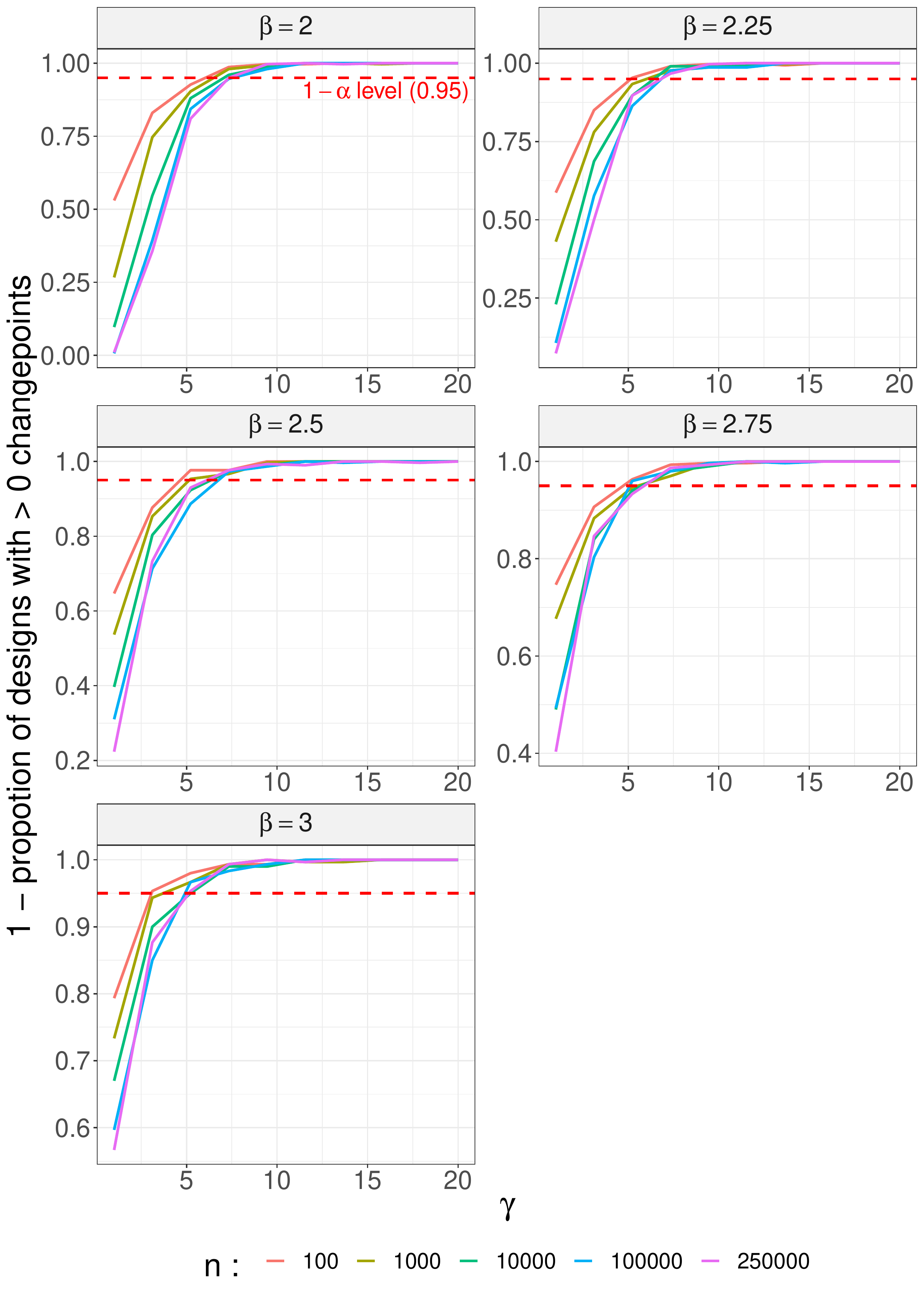}
    \caption{\textbf{Proportion of stationary Gaussian process replicates on which \PhiFPOP\ returns at least 1 changepoint ($R_{>0}$).} $R_{>0}$ is computed for a series of $\gamma$, $\beta$, and profile lengths (see \textit{Design of Simulations} in section \ref{calibration}).}
\end{figure}

\section{\PhiFPOP{} Speed Benchmark} \label{supp_speed_benchmark}

\paragraph{Sampling Strategies} We compared the runtime of \PhiFPOP{} for various sampling strategies (see section \ref{sec:subsampling}).
We tested sampling 1, 2, 3 and all future changes. We call these strategies respectively rand 1, rand 2, rand 3 and all. We tested them on the simulation described in section \ref{speed_bench}.

It can be seen on the Figure \ref{fig:subsamplingspeed} that sampling 1 future change uniformaly at random is the fastest for all true number of changes and $n=10^5.$

\begin{figure}[H]
\centering
\includegraphics[scale=0.36]{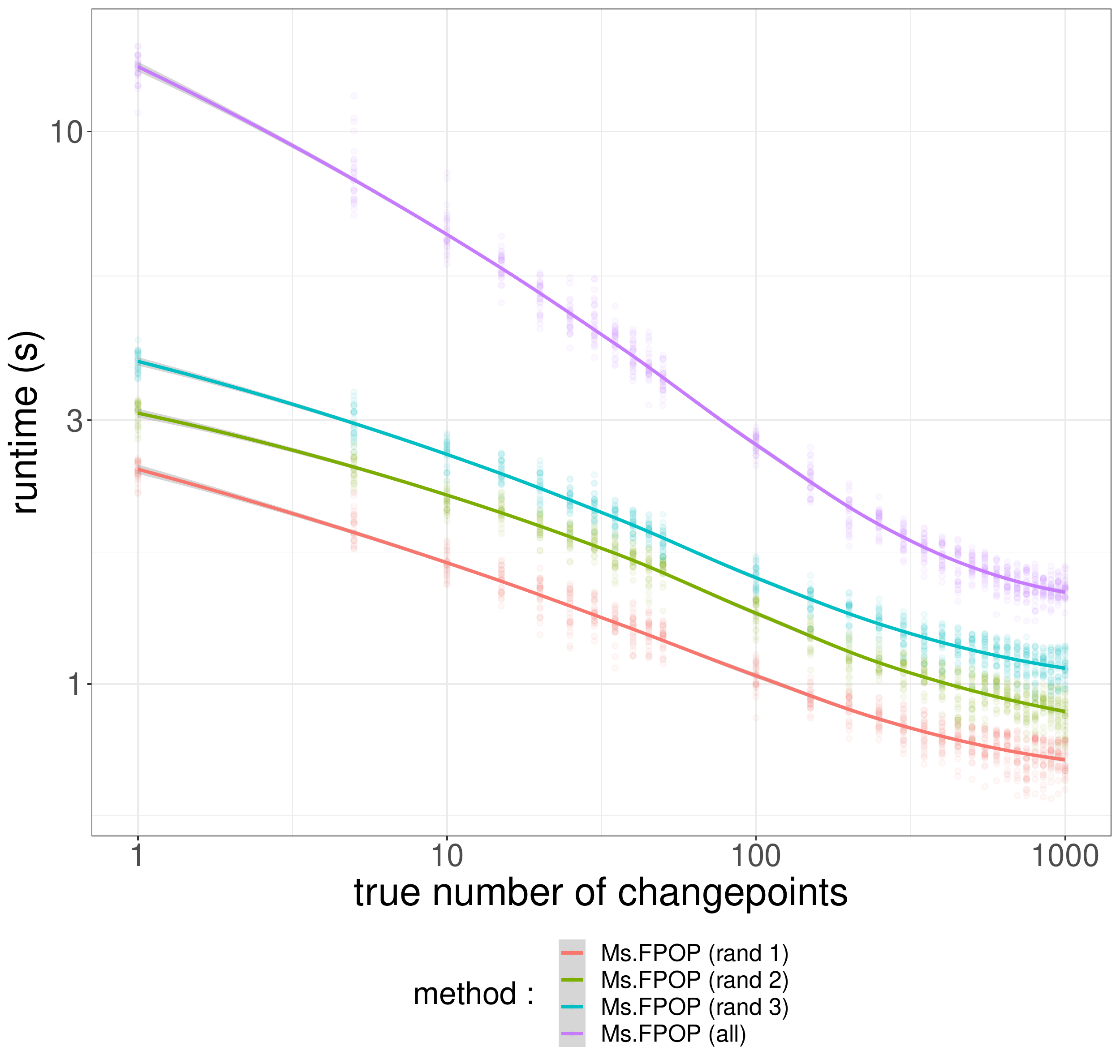}
\caption{\textbf{Runtimes as a function of the true number of changepoints.} We timed strategies consisting in randomly sampling one, two, three, four or all future candidates on profiles of length $n=10^5$ with varying number of true changepoints $D$ (see \textit{Design of Simulations} in section \ref{speed_bench}) on an Intel Core i7-10810U CPU @ 1.10GHzx12 computer. We observe that, marginaly to $D$, randomly sampling 1 future candidate (\PhiFPOP\ rand 1) is faster than randomly sampling more than one future candidates, in particular all future candidates.}\label{fig:subsamplingspeed}
\end{figure}

\paragraph{Larger Profile Lengths $(n=10^6)$} Figure \ref{fig:speed1M} is obtained as explained in section \ref{speed_bench}, with $n=10^6$ and $D \in \{1, 500, 1000, 1500, 2000, 2500, 3000, 3500, 4000, 4500,$\\$ 5000, 5500, 6000, 6500, 7000, 7500, 8000, 8500, 9000, 9500, 10000\}$.

\begin{figure}[H]
\centering
\includegraphics[scale=0.36]{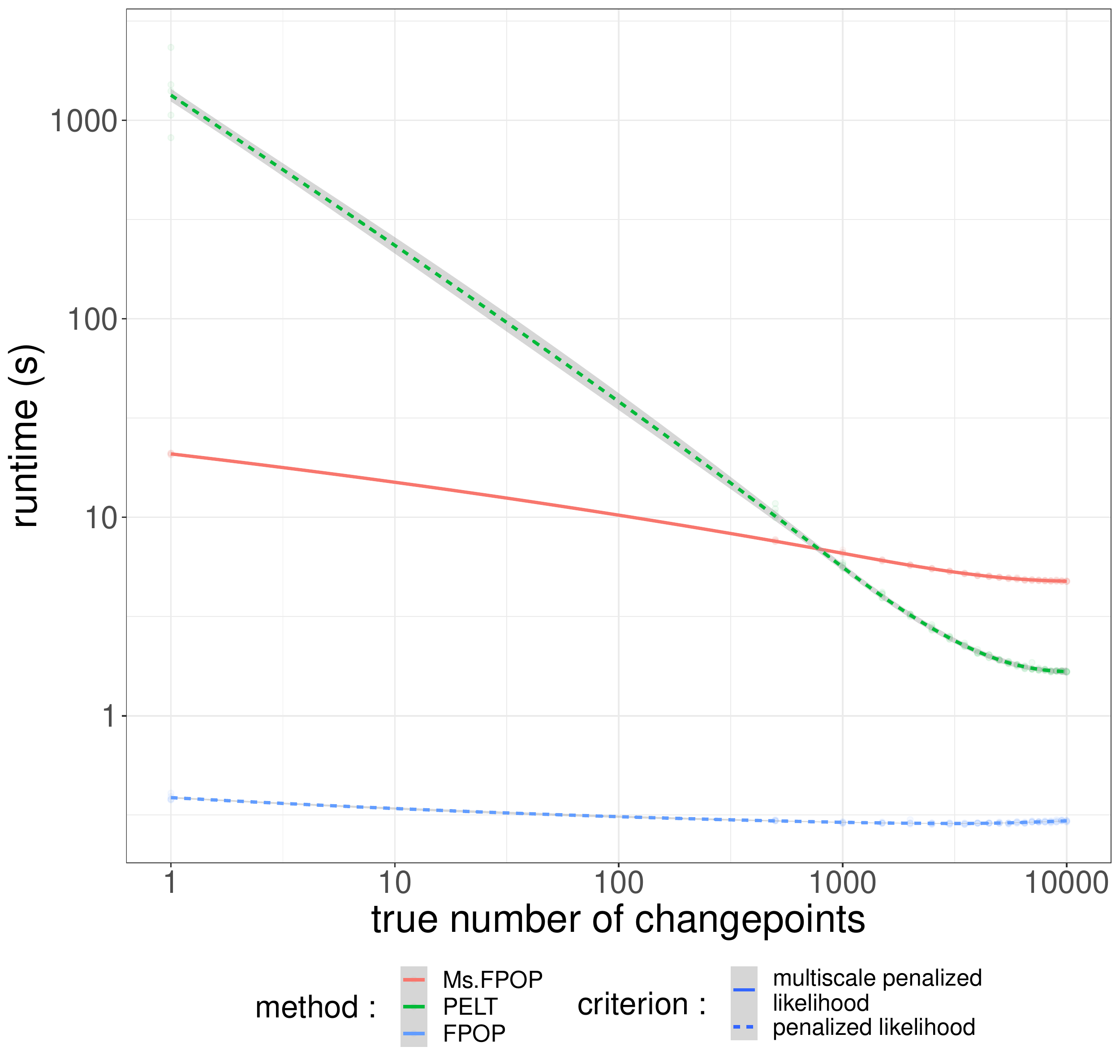}
\caption{\textbf{Runtimes as a function of the true number of changepoints.} We timed \PELT, \FPOP, \PhiFPOP{} on profiles of length $n=10^6$ with varying number of true changepoints $D$ (see \ref{supp_speed_benchmark}) on an Intel Core i7-10810U CPU @ 1.10GHzx12 computer.}\label{fig:speed1M}
\end{figure}

\section{\FPOP\ vs \PhiFPOP\ : Simulations on Hat Profiles} \label{supp:hat_simu}

Figure \ref{fig:supphat_simu} is obtained as explained in section \ref{text:hat_simu}.
\begin{figure}[h!]
    \hspace{-1cm}\includegraphics[scale=0.42]{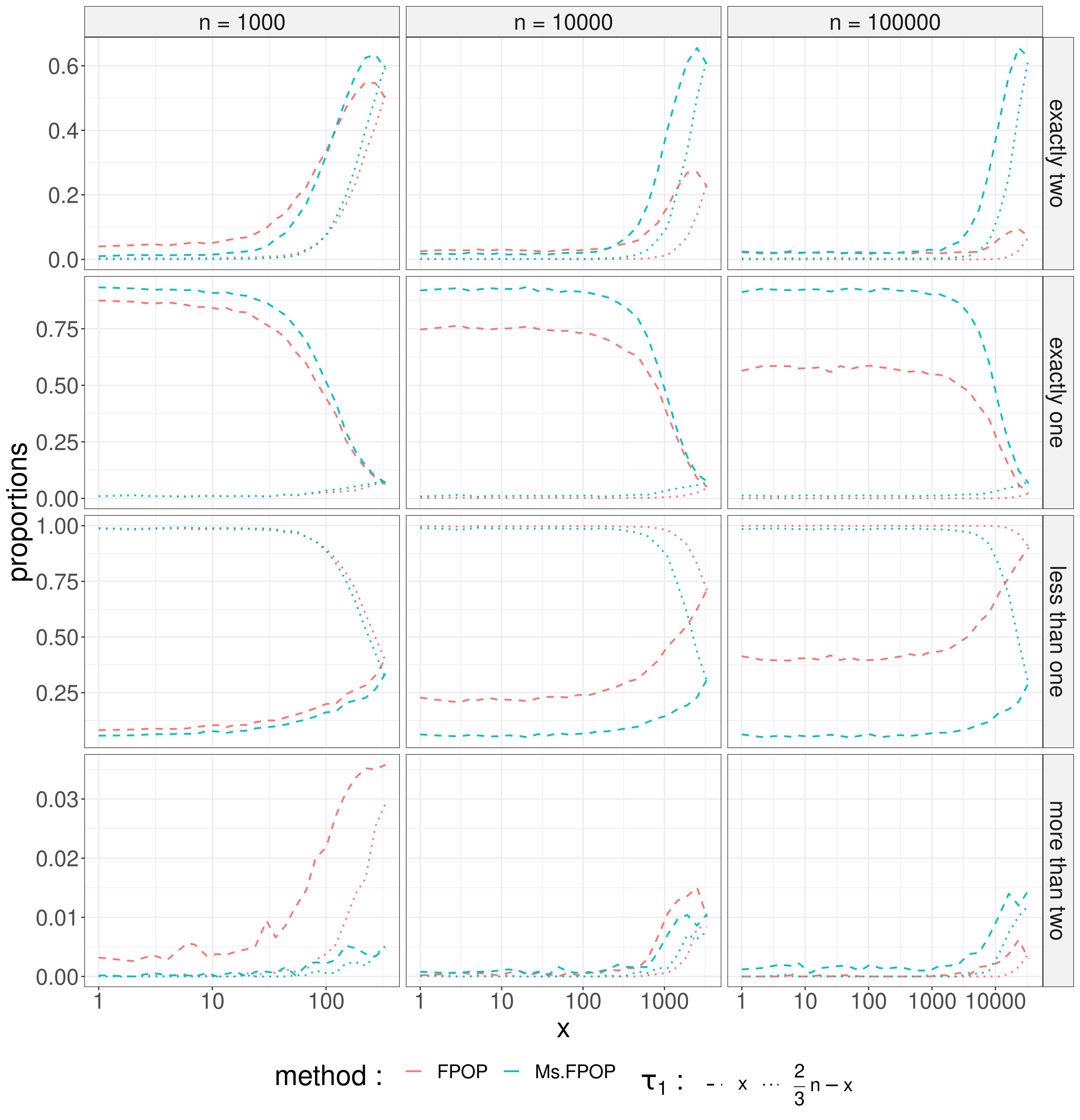}
    \caption{\textbf{\PhiFPOP\ increases the probability of finding well spread changepoints on \textit{hat profiles}.} The proportion of replicates for which \PhiFPOP\ and \FPOP\ return two changepoints, one changepoint, zero changepoint, and more than two changepoints are computed for varying $\tau_1 \in [1,\lfloor\frac{2}{3}n\rfloor-1]$ and $n \in [10^3, 10^4, 10^5]$ on the hat-like profiles (see \textit{Design of Simulations} in section \ref{text:hat_simu}).}\label{fig:supphat_simu}
\end{figure}

\section{\FPOP\ vs \PhiFPOP\ : Simulations on Step Profiles}
\label{supp:step_simu}

\begin{figure}[h!]
    \centering\includegraphics[scale=0.55]{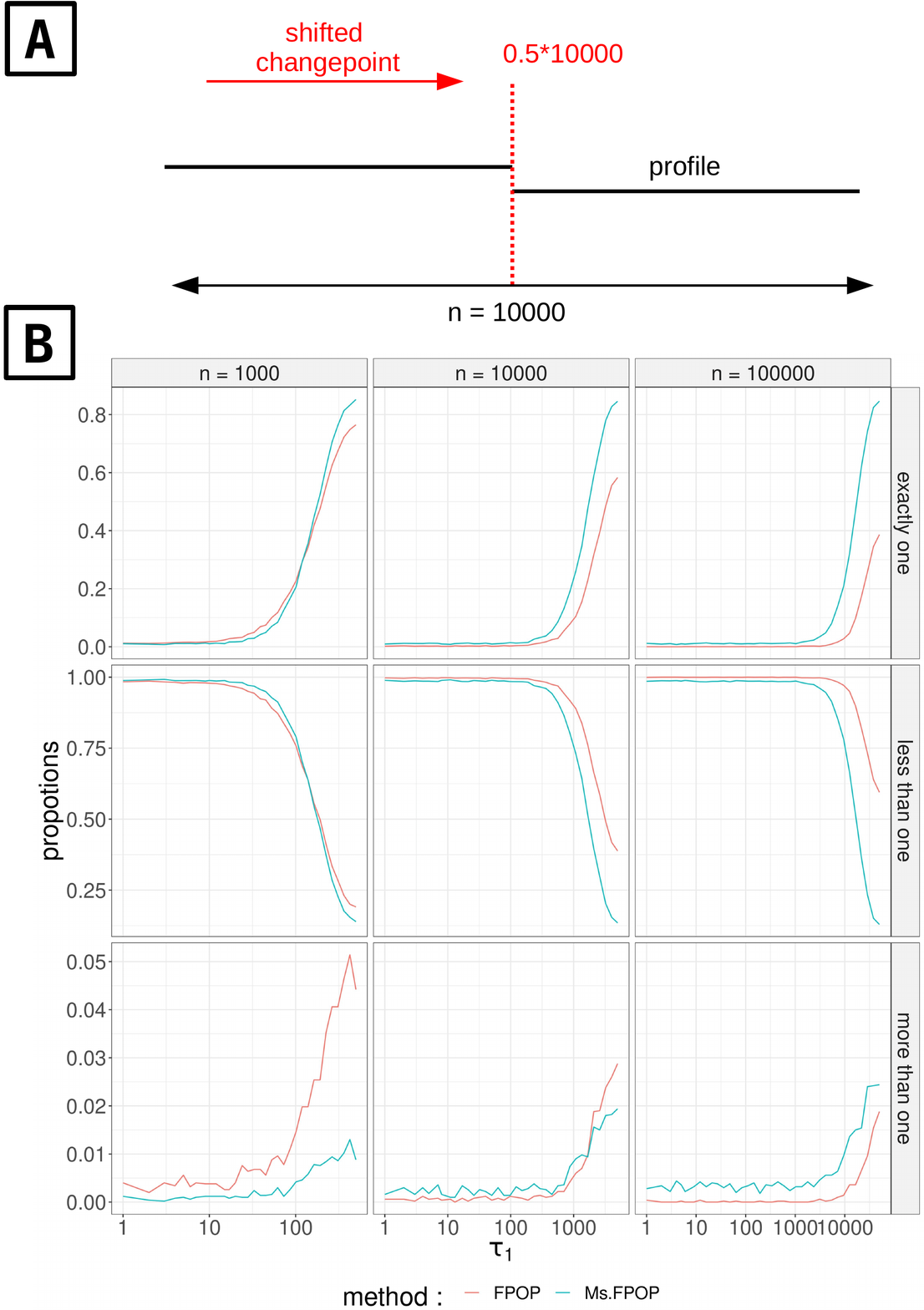}
    \vspace{-0.5cm}\caption{\textbf{\PhiFPOP\ increases the probability of finding well spread changepoint on \textit{step profiles}.} - (A) denoised profile with 1 changepoint. The first and unique changepoint ($\tau_1$) varies on the interval $[1,\lfloor\frac{n}{2}\rfloor$] (40 positive integers evenly spaced on the $\log$ scale). The difference in mean between the second and the first segment is set to $\sqrt{\frac{70}{n}}$. An \textit{iid} Gaussian noise of variance one is then added. - (B) The proportion of replicates on which \PhiFPOP\ and \FPOP\ return one changepoint, less than one changepoint and more than one changepoints are computed for varying profile lengths ($n \in {10^3, 10^4, 10^5}$) and $\tau_1$.}\label{fig:suppstep_simu}
\end{figure}

\section{\FPOP\ vs \PhiFPOP\ vs \textbf{MOSUM} : Simulations on Several Scenarios of Gaussian Signals (segments length $>300$)}
\label{supp:other_simu_min300}

Figures \ref{fig:all_scenarios_min300}, \ref{fig:AE_perc_K_min300}, \ref{fig:AE_perc_ARI_min300}, \ref{fig:AE_perc_MSE_min300} were obtained as explained in section \ref{extended_sim} when considering the benchmark in \cite{Fearnhead2020}. 
On these simulations  a large portion of the segments have a length under 100.

\begin{figure}[h!]
    \centering\includegraphics[scale=0.42]{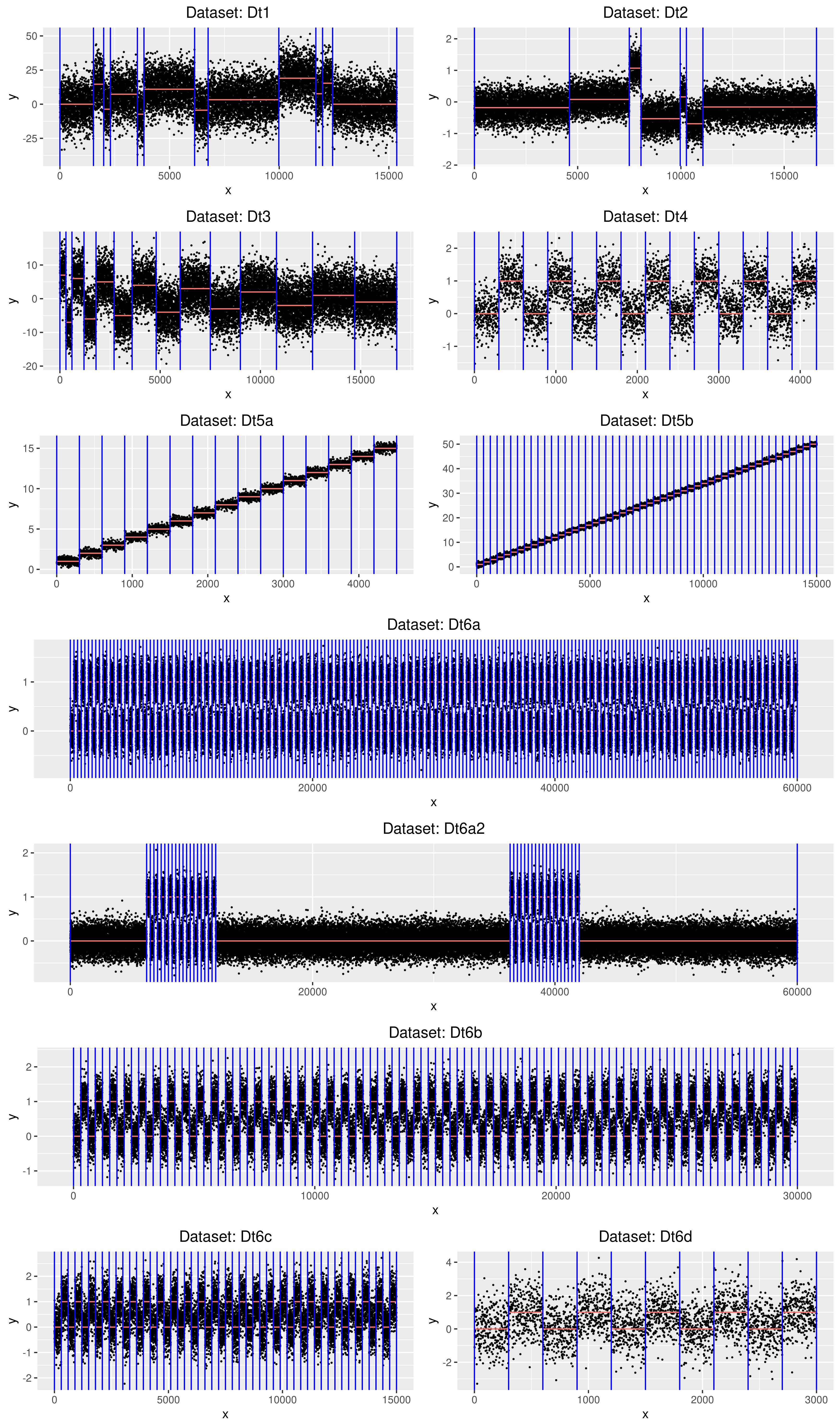}
    \vspace{-0.5cm}\caption{\textbf{Simulated scenarios of Gaussian signals with minimum segments length equal to 300.} All scenarios have been simulated following the protocol written by Fearnhead \textit{et al} 2020. The length of each segment are scaled so that, in each profile, all segments contain at least 300 datapoints.}
    \label{fig:all_scenarios_min300}
\end{figure}

\begin{figure}[h!]
    \centering\includegraphics[scale=0.33]{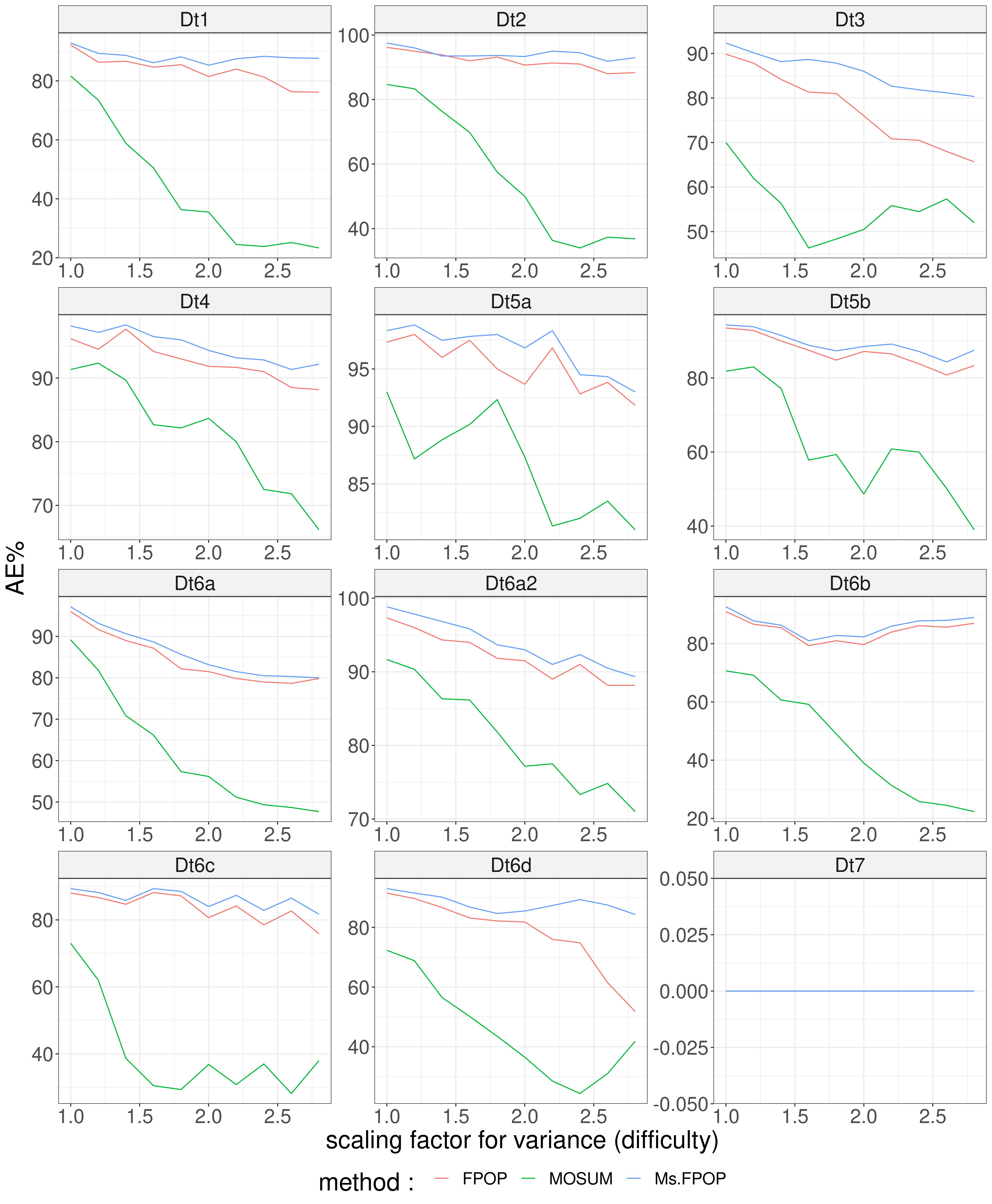}
    \vspace{-0.5cm}\caption{\textbf{AE\% as a function of the scaling factor for the variance (comparison criterion : ARI).} The average number of times a method is at least as good as other methods in terms of ARI is computed for \FPOP, \PhiFPOP, and \textbf{MOSUM} on different scenarios of \textit{iid} Gaussian signals and varying $\sigma^2$. The smallest segment length is greater or equal to 300 (see \textit{Design of Simulations}). Each panel stands for the results on one scenario. Corresponding profiles can be viewed in \ref{supp:other_simu_min300}.}
    \label{fig:AE_perc_ARI_min300}
    \label{fig:my_label}
\end{figure}

\begin{figure}[h!]
\hspace{-1.2cm}\includegraphics[scale=0.33]{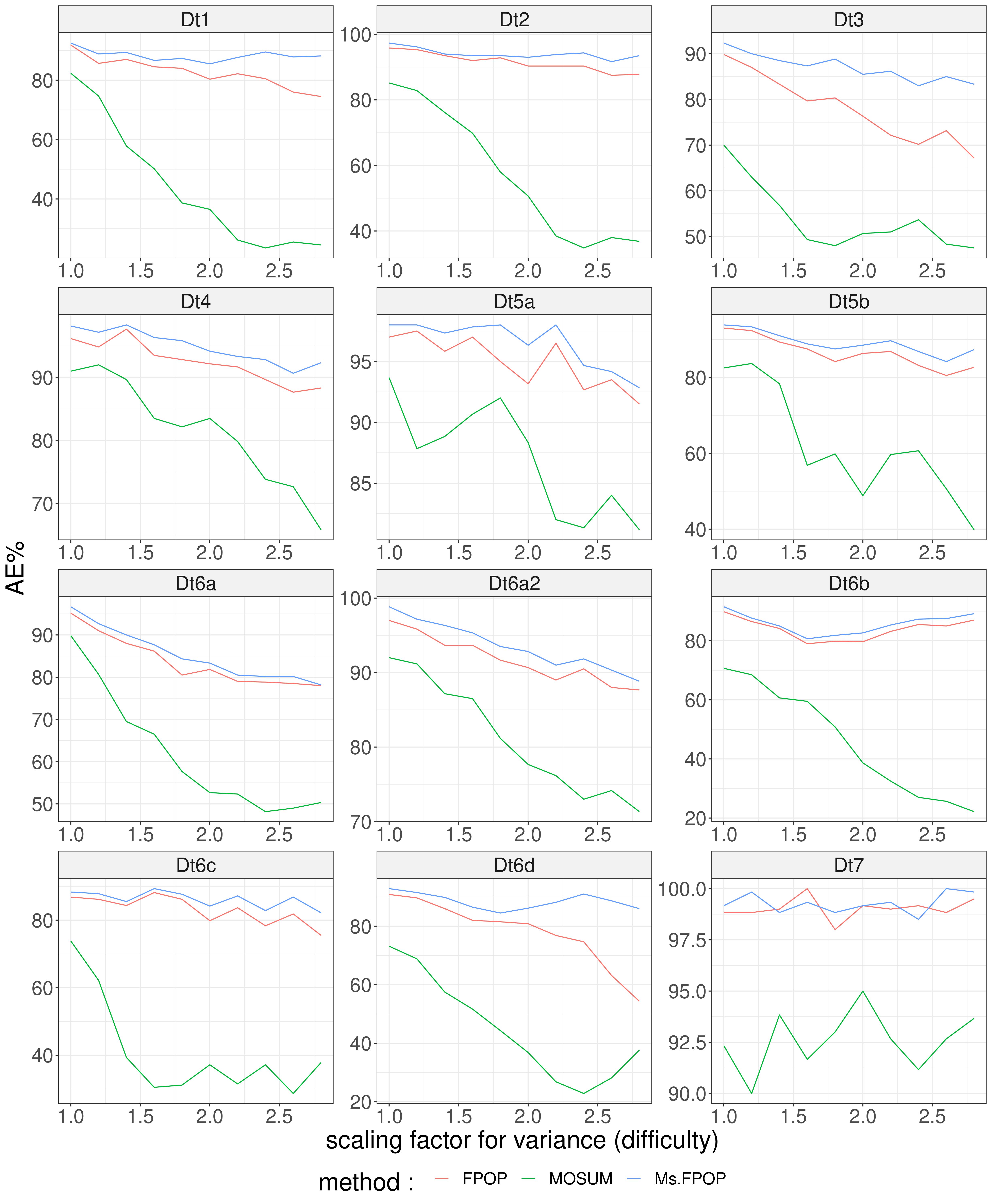}
    \centering\caption{\textbf{AE\% as a function of the scaling factor for the variance (comparison criterion : MSE).} The average number of times a method is at least as good as other methods in terms of MSE is computed for \FPOP, \PhiFPOP, and \textbf{MOSUM} on different scenarios of \textit{iid} Gaussian signals and varying $\sigma^2$. The smallest segment length is greater or equal to 300 (see \textit{Design of Simulations}). Each panel stands for the results on one scenario. Corresponding profiles can be viewed in \ref{supp:other_simu_min300}.}
    \label{fig:AE_perc_MSE_min300}
\end{figure}

\section{\FPOP\ vs \PhiFPOP\ vs \textbf{MOSUM} : Simulations on Several Scenarios of Gaussian Signals} \label{supp:other_simu_default}

Figures \ref{fig:all_scenarios}, \ref{fig:AE_perc_K}, \ref{fig:AE_perc_ARI}, \ref{fig:AE_perc_MSE} were obtained as explained in section \ref{extended_sim} when considering an extension of the benchmark in \cite{Fearnhead2020}.
Based on the initial scenarios we simulated another set of profiles in which segments length are multiplied so that each of segments contain at least 300 datapoints.

\begin{figure}[h!]
    \centering\includegraphics[scale=0.42]{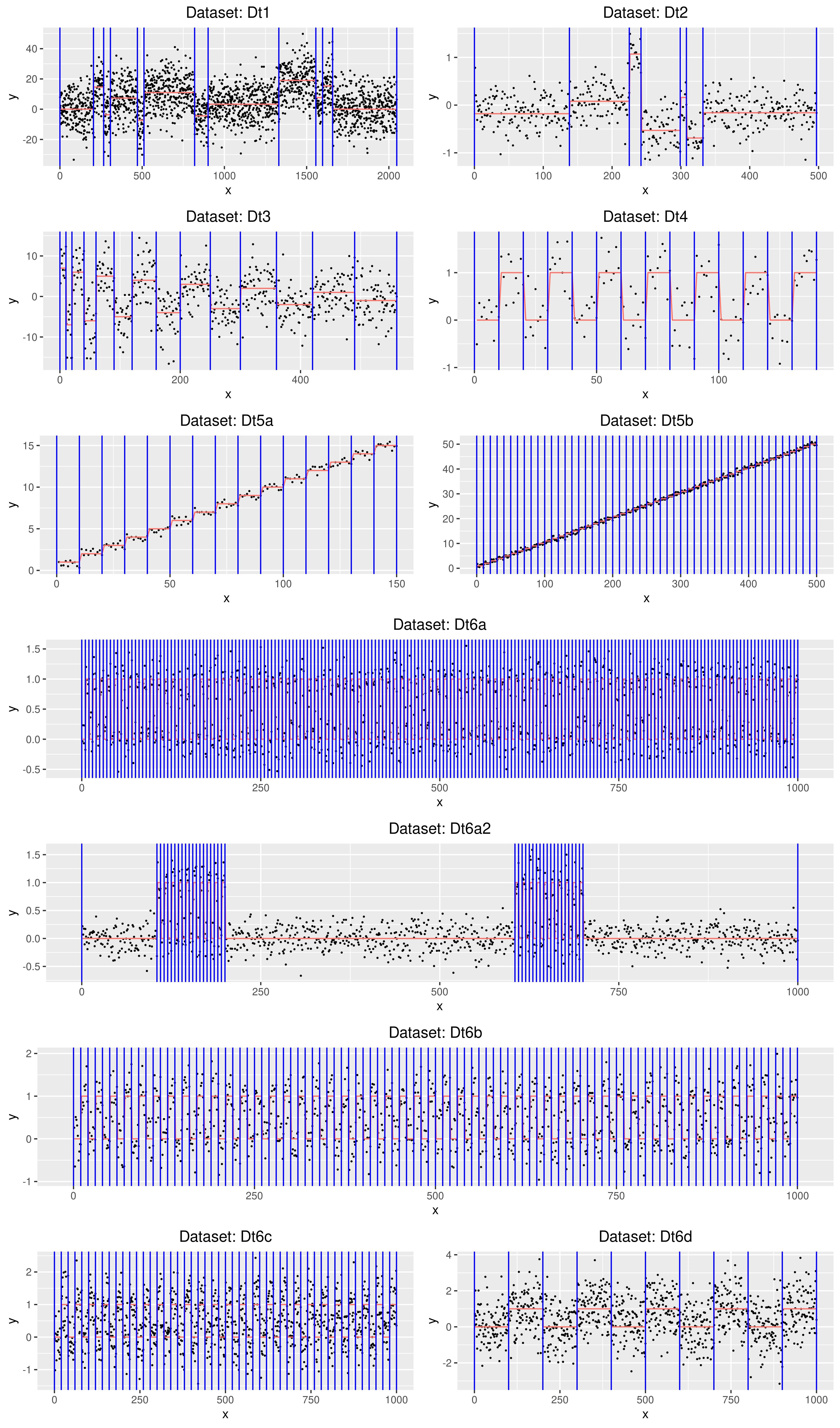}
    \vspace{-0.5cm}\caption{\textbf{Simulated scenarios of Gaussian signals.} All scenarios have been simulated following the protocol written by Fearnhead \textit{et al} 2020.}
    \label{fig:all_scenarios}
\end{figure}

\begin{figure}[h!]
\hspace{-1.2cm}\includegraphics[scale=0.33]{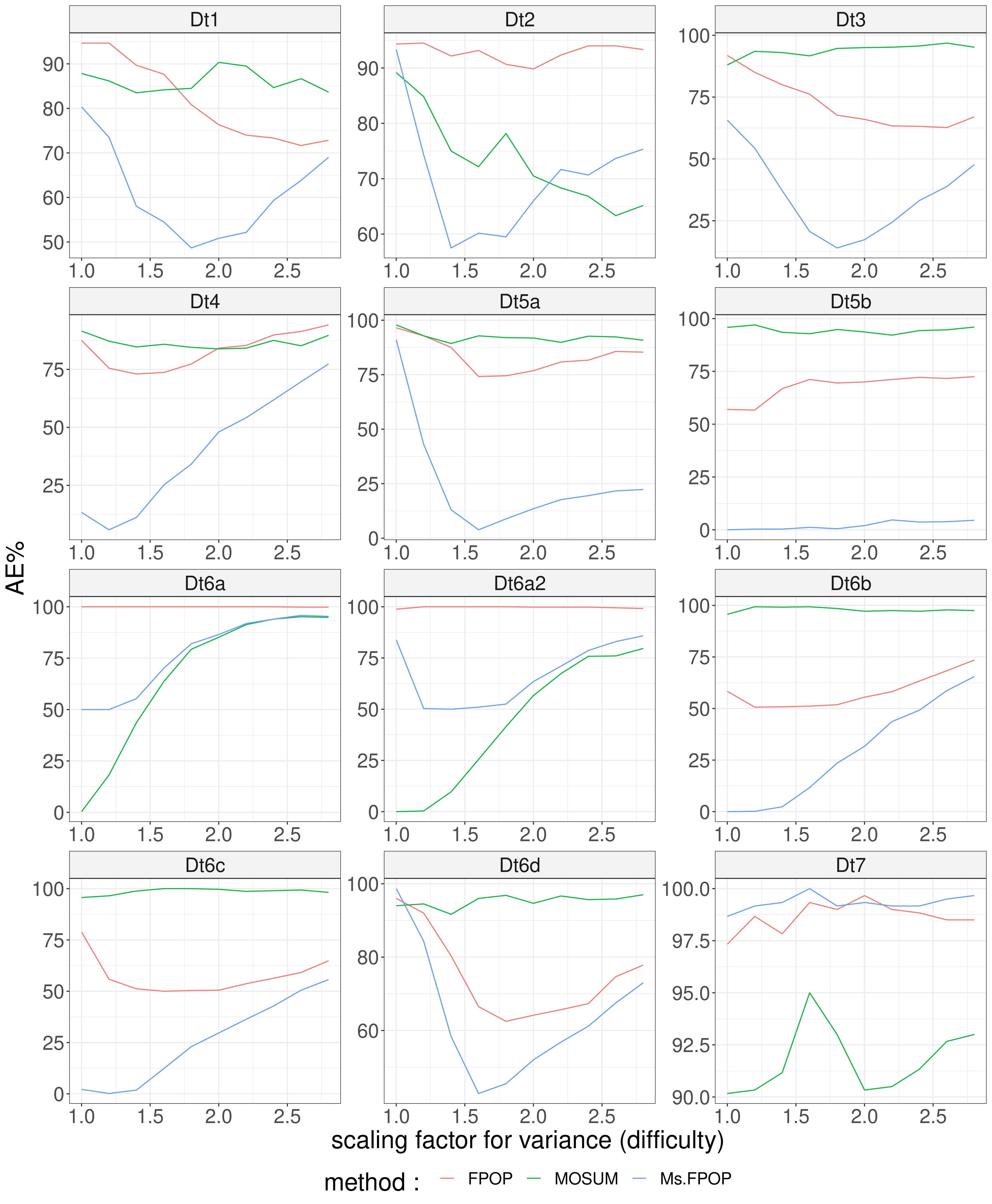}
    \centering\caption{\textbf{AE\% as a function of the scaling factor for the variance (comparison criterion : $\Delta_D$).} The average number of times a method is at least as good as other methods in terms of $\Delta_D$ is computed for \FPOP, \PhiFPOP, and \textbf{MOSUM} on different scenarios of \textit{iid} Gaussian signals and varying $\sigma^2$. (see \textit{Design of Simulations}). Each panel stands for the results on one scenario. Corresponding profiles can be viewed in \ref{supp:other_simu_default}.}
    \label{fig:AE_perc_K}
\end{figure}

\begin{figure}[h!]
    \centering\includegraphics[scale=0.33]{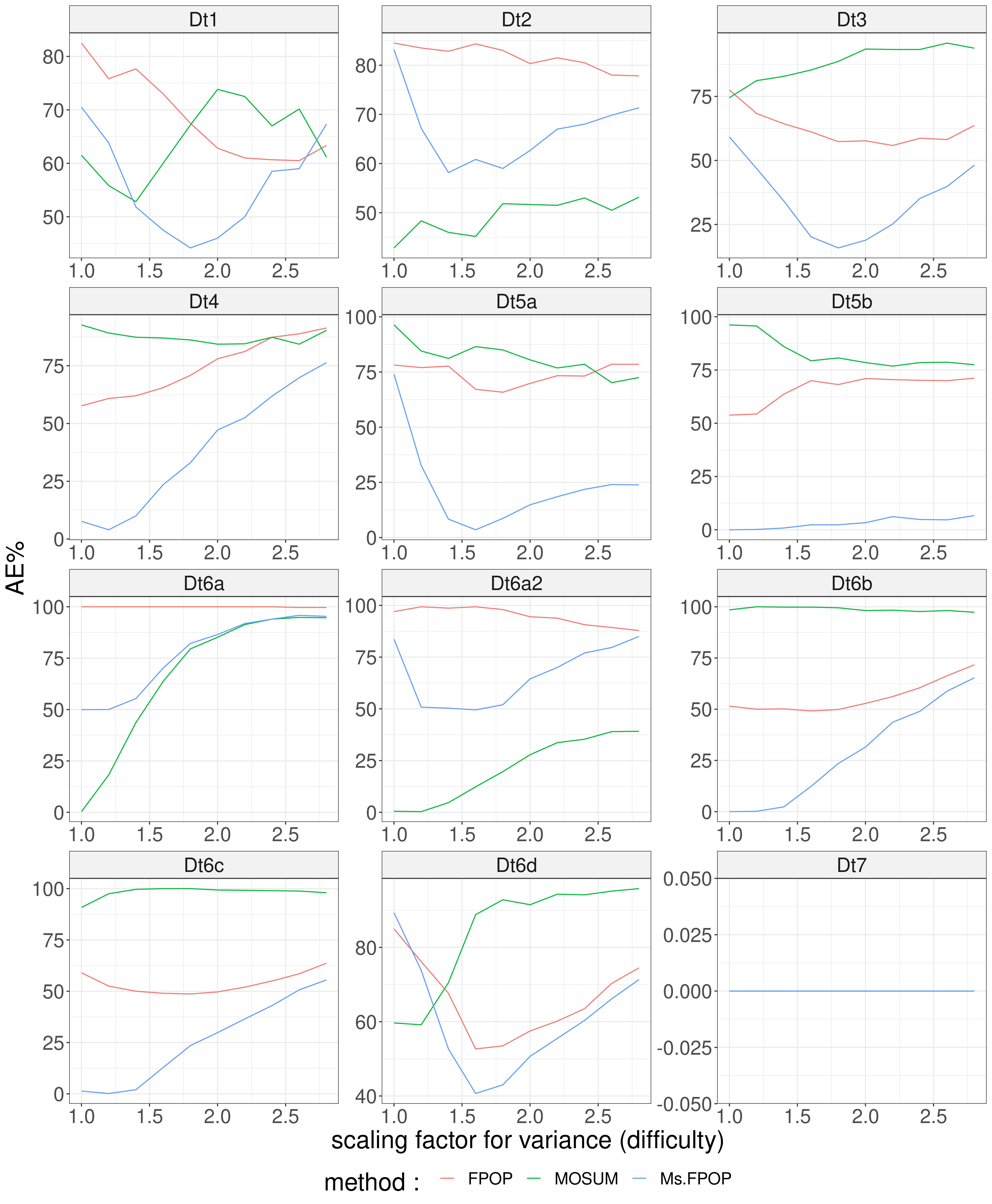}
    \vspace{-0.5cm}\caption{\textbf{AE\% as a function of the scaling factor for the variance (comparison criterion : ARI).} The average number of times a method is at least as good as other methods in terms of ARI is computed for \FPOP, \PhiFPOP, and \textbf{MOSUM} on different scenarios of \textit{iid} Gaussian signals and varying $\sigma^2$. (see \textit{Design of Simulations}). Each panel stands for the results on one scenario. Corresponding profiles can be viewed in \ref{supp:other_simu_default}.}
    \label{fig:AE_perc_ARI}
\end{figure}

\begin{figure}[h!]
\centering\includegraphics[scale=0.33]{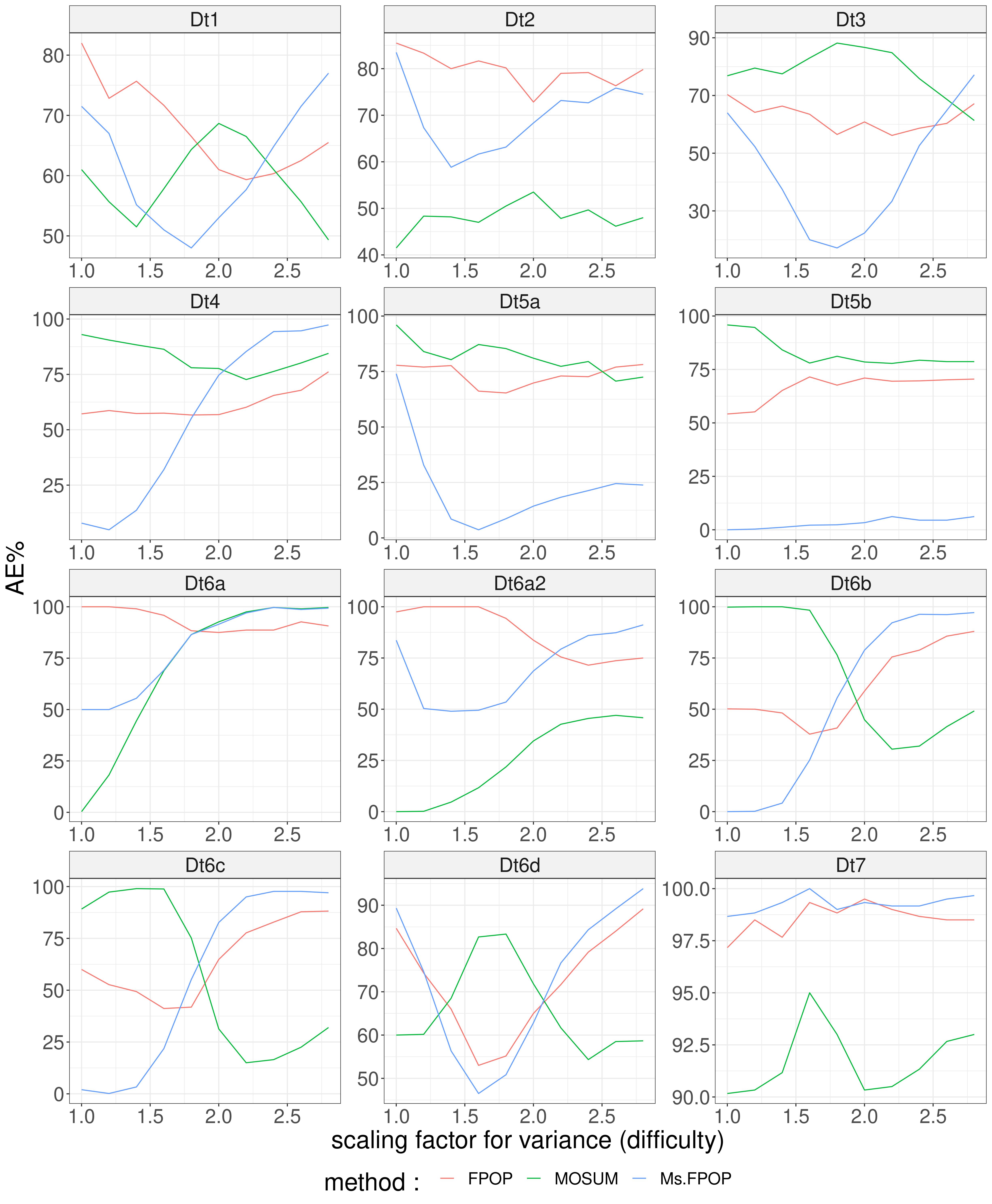}
    \vspace{-0.5cm}\caption{\textbf{AE\% as a function of the scaling factor for the variance (comparison criterion : MSE).} The average number of times a method is at least as good as other methods in terms of MSE is computed for \FPOP, \PhiFPOP, and \textbf{MOSUM} on different scenarios of \textit{iid} Gaussian signals and varying $\sigma^2$. (see \textit{Design of Simulations}). Each panel stands for the results on one scenario. Corresponding profiles can be viewed in \ref{supp:other_simu_default}.}
    \label{fig:AE_perc_MSE}
\end{figure}

\bibliography{references}
\end{document}